\documentclass[10pt,english]{IEEEtran}
\usepackage[T1]{fontenc}
\usepackage[latin9]{inputenc}
\usepackage{calc}
\usepackage{units}
\usepackage{amsmath}
\usepackage{amssymb}
\usepackage{graphicx}

\makeatletter
  \newtheorem{assumption}{Assumption}
  \newtheorem{remrk}{Remark}
  \newtheorem{definitn}{Definition}
  \newtheorem{lemma}{Lemma}
  \newtheorem{thm}{Theorem}
  \newtheorem{cor}{Corollary}

\usepackage{cite}

\author{Xiongbin~Rao,~\IEEEmembership{Student~Member,~IEEE} and~Vincent~K.~N.~Lau,~\IEEEmembership{Fellow,~IEEE}
\thanks{The authors are with the Department of Electronic and Computer Engineering (ECE), the Hong Kong University of Science and Technology (HKUST), Hong Kong (e-mail: \{xrao,eeknlau\}@ust.hk).}%
 }

\makeatother

\usepackage{babel}
\begin{document}

\title{Distributed Fronthaul Compression and Joint Signal Recovery in Cloud-RAN}
\maketitle
\begin{abstract}
The cloud radio access network (C-RAN) is a promising network architecture
for future mobile communications, and one practical hurdle for its
large scale implementation is the stringent requirement of high capacity
and low latency \emph{fronthaul} connecting the \emph{distributed}
remote radio heads (RRH) to the centralized baseband pools (BBUs)
in the C-RAN. To improve the scalability of C-RAN networks, it is
very important to take the \emph{fronthaul loading }into consideration
in the signal detection, and it is very desirable to \emph{reduce
the fronthaul loading} in C-RAN systems. In this paper, we consider
uplink C-RAN systems and we propose a \emph{distributed fronthaul
compression }scheme at the distributed RRHs\emph{ }and\emph{ }a \emph{joint
recovery} algorithm at the BBUs by deploying the techniques of \emph{distributed
compressive sensing} (CS). Different from conventional distributed
CS, the CS problem in C-RAN system needs to incorporate the underlying
effect of multi-access fading for the end-to-end recovery of the transmitted
signals from the users. We analyze the performance of the proposed
end-to-end signal recovery algorithm and we show that the aggregate
measurement matrix in C-RAN systems, which contains both the distributed
fronthaul compression and multiaccess fading, can still satisfy the
\emph{restricted isometry property} with high probability. Based on
these results, we derive tradeoff results between the uplink capacity
and the fronthaul loading in C-RAN systems. \end{abstract}
\begin{keywords}
Cloud radio access network (C-RAN), distributed fronthaul compression,
joint signal recovery, active user detection, restricted isometry
property (RIP). \vspace{-0.2cm}

\end{keywords}

\section{Introduction}

Mobile data has been growing enormously in recent years, and to meet
the increasing demand, researchers have proposed various advanced
technologies to enhance the spectrum efficiency of wireless systems.
In traditional cellular networks, data detection is done locally at
each base station (BS), and a significant portion of the transmit
power is used to overcome the path loss as well as interference from
other user equipment (UEs). As a result, the cellular network is interference
limited. Recently, a lot of interference mitigation techniques which
exploit multi-cell cooperation (Cooperative Multi-Point (CoMP) processing)
have been presented \cite{tolli2008cooperative,wang2010cooperative}.
However, these techniques have stringent synchronization and backhaul
capacity and latency demands to exchange the channel fading states
and payload among different BSs. To meet these stringent requirements,
a new network architecture, namely the \emph{cloud radio access network}
(C-RAN), has been proposed \cite{mobile2011c} and it has received
lots of research interest recently \cite{chanclou2013optical,bernardos2013challenges}.
Figure \ref{fig:Illustration-of-Cloud} illustrates a C-RAN system
which consists of a number of remote radio heads (RRHs), a pool of
baseband processing units (BBUs) in the cloud as well as a high bandwidth,
low latency optical transport network (fronthaul \cite{chanclou2013optical})
connecting the RRHs to the BBU cloud. The RRHs consist of simple low
power antennas and RF components, and the baseband processing is centralized
at the BBU. Such centralization offers effective CoMP interference
mitigation and resource pooling gain in the cloud as well as BS virtualization
\cite{mobile2011c}. 

Despite various attractive features of the C-RAN, one practical hurdle
for its large scale implementation is the stringent requirement of
high capacity and low latency fronthaul connecting the RRHs to the
BBUs in the cloud. Due to the huge numbers of RRHs involved, the fronthaul
is one of the most dominant cost components in C-RAN and it is very
important to \emph{reduce the fronthaul loading} in order to improve
the scalability of C-RAN networks \cite{bernardos2013challenges}.
There are various approaches in the literature to reduce the backhaul
loading in traditional cellular networks. In \cite{ng2010linear,hong2012joint},
partial CoMP schemes via \emph{sparse precoding} are proposed to reduce
the backhaul consumption in cellular networks. In \cite{zhou2014optimized,park2013robust,park2013robuster,del2009distributed},
distributed source coding strategies are used to compress the fronthaul
signals in the uplink of C-RAN. By maximizing the weighted sum rate
with respect to the fronthaul codebooks subject to fronthaul capacity
constraint \cite{zhou2014optimized,park2013robust,park2013robuster,del2009distributed},
the fronthaul loading in C-RAN can be effectively reduced. However,
in these works \cite{ng2010linear,hong2012joint,park2013robust,park2013robuster,del2009distributed,zhou2014optimized},
the uplink signal sparsity structure is ignored and the potential
benefits brought by exploiting the uplink signal sparsity is not considered.
In practice, the UEs' uplink packets might be sparse due to bursty
transmissions of the UEs in delay-sensitive services \cite{yun2008analysis}
or the random access of the UEs (e.g., future machine-type communications
\cite{3gppM2M} involve massive machine-type UEs with low-latency
and bursty data \cite{3gppM2M,hasan2013random,Laya2014RACH}). As
such, there is huge potential to exploit the uplink user sparsity
to further reduce the fronthaul loading. In \cite{cdmamulde2009hao,schepker2011sparse,zhu2011exploiting},
the authors consider sparse uplink users in CDMA systems and compressive
sensing (CS) recovery is deployed to improve the multi-user detection
performance. However, these techniques \cite{cdmamulde2009hao,schepker2011sparse,zhu2011exploiting}
cannot be extended to the C-RAN scenarios.

In this paper, we are interested in fronthaul distributed compression
using \emph{distributed compressive sensing} and\emph{ recovery} techniques
by exploiting the signal sparsity in uplink C-RAN systems. Specifically,
each RRH compresses the uplink signal locally and sends it to the
BBU pools in the cloud. The BBU pools then jointly recover the transmitted
signals from the UEs based on the compressed uplink signals from various
RRHs without knowledge of the statistics of these uplink signals.
In the literature, distributed compressive sensing and recovery, in
which a set of jointly sparse signals are compressively sampled and
then jointly recovered, have been studied in \cite{baron2005distributed}.
However, these existing results cannot be directly applied to the
C-RAN fronthaul compression problem and there are several first order
technical challenges involved. 
\begin{itemize}
\item \textbf{Distributed Fronthaul Compression and Joint Data Recovery
with Multi-access Fading.} Classical distributed CS \cite{baron2005distributed}
concerns the joint recovery of signals that are compressed distributively
at each sensor. However, these existing techniques cannot be directly
applied to the C-RAN scenarios. While in C-RAN systems, the RRH locally
compresses the received signals, these locally compressed signals
are aggregations of the multi-user signals transmitted by different
UEs over \emph{multi-access fading,} as illustrated in Figure \ref{fig:Joint-signal-recovery}.
The target of the joint recovery in the C-RAN scenarios are the \emph{transmitted
signals} from the UEs rather than the locally received signals in
the RRH (as in conventional distributed CS \cite{baron2005distributed}),
as illustrated in Figure \ref{fig:Joint-signal-recovery}. As a result,
a new distributed CS problem formulation and recovery that incorporate
the effects of multi-access fading in C-RAN, is needed. 
\item \textbf{Robust CS Recovery Conditions with Multi-access Fading.} In
the literature, the restricted isometry property (RIP) \cite{candes2005decoding}
is commonly adopted to provide a sufficient condition for robust CS
recovery \cite{candes2008introduction}, and it is highly non-trivial
to establish a sufficient condition for the RIP of the associated
measurement matrix. Note that the RIP characterizations are known
for sub-sampled Fourier transformation matrices and random matrices
with i.i.d. sub-Gaussian entries \cite{candes2008introduction}. However,
due to the complicated multi-access fading channels between the UE
and the RRHs in C-RAN systems, the associated \emph{aggregate} CS
measurement matrix does not belong to any of the known measurement
matrix structures, and hence, conventional results about the RIP condition
cannot be applied and a new characterization of the sufficient conditions
for robust CS recovery (embracing multi-access fading) will be needed
for C-RAN. 
\item \textbf{Tradeoff Analysis between C-RAN Performance and the Fronthaul
Loading.} Besides applying CS in the signal recovery to achieve fronthaul
compression, it is also very important to quantify the closed-form
tradeoff between the capacity of the C-RAN and the fronthaul loading
at each RRH. This tradeoff result will be important to reveal design
insights and guidelines on the dimensioning of the fronthaul (e.g.,
how large a fronthaul is needed to achieve a certain capacity target
in the C-RAN). However, the closed form performance analysis will
be very challenging. 
\end{itemize}

\begin{figure}
\begin{centering}
\includegraphics[scale=0.5]{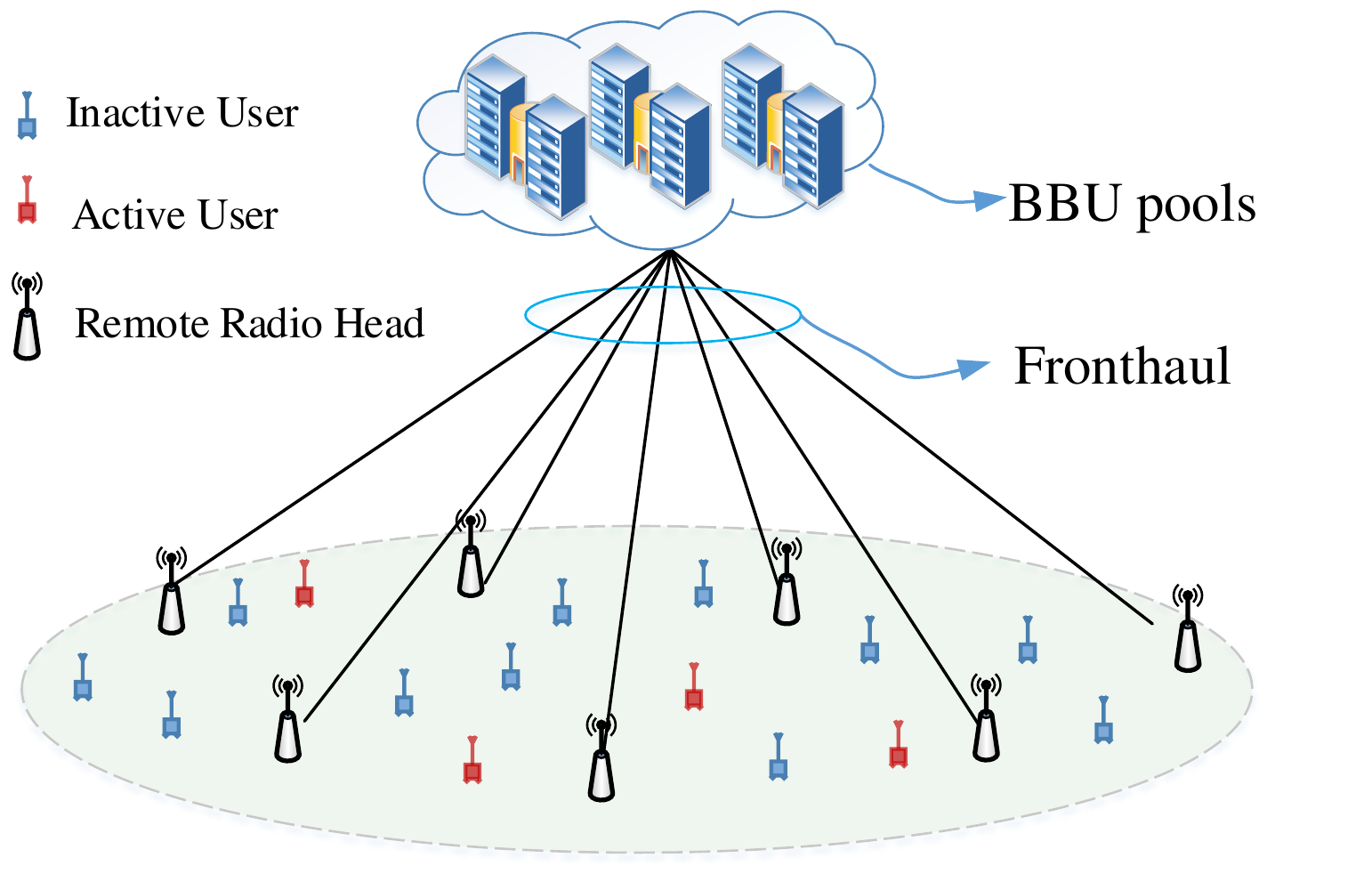}
\par\end{centering}

\protect\caption{\label{fig:Illustration-of-Cloud}Illustration of a C-RAN system.
The active user and inactive user, refer to the UEs that are transmitting
signal and keeping silent, respectively, in the considered time slot. }
\end{figure}

In this paper, we shall address the above challenges. We first introduce
the uplink C-RAN model and propose a distributed fronthaul compression
scheme at the RRHs of the C-RAN. We propose a joint signal recovery
scheme at the BBU pools, which exploits the underlying signal sparsity
of the UEs and the multi-access fading effects between the UEs and
the RRHs. Based on that, we analyze the associated CS measurement
matrix and show that the RIP condition \cite{candes2005decoding}
can be satisfied with high probability under some mild conditions.
Furthermore, we characterize the achievable uplink capacity in terms
of the compression rate on the fronthaul. From the results, we draw
simple conclusions on the tradeoff relationship between the uplink
C-RAN performance and the fronthaul loading in the C-RAN. Finally,
we verify the effectiveness of the proposed distributive fronthaul
compression scheme via simulations. 

\textit{Notation}s: Uppercase and lowercase boldface letters denote
matrices and vectors respectively. The operators $(\cdot)^{T}$, $(\cdot)^{*}$,
$(\cdot)^{H}$, $(\cdot)^{\dagger}$, $\#|\cdot|$, $\textrm{tr}\left(\cdot\right)$,
$O(\cdot)$ and $\textrm{Pr}(\cdot)$ are the transpose, conjugate,
conjugate transpose, Moore-Penrose pseudoinverse, cardinality, trace,
big-O notation, and probability operator respectively; $\mathbf{A}(q,p)$
and $\mathbf{a}(l)$ denote the $(q,p)$-th entry of $\mathbf{A}$
and the $l$-th entry of $\mathbf{a}$ respectively; $\mathbf{A}_{\Omega}$
and and $\mathbf{a}_{\Omega}$ denote the sub-matrices formed by collecting
the columns of $\mathbf{A}$ and sub-vector formed by collecting the
entries of $\mathbf{a}$, respectively, whose indexes are in set $\Omega$;
$||\mathbf{a}||_{q}$ is the $l_{q}$-norm of the vector $\mathbf{a}$
and is defined as $||\mathbf{a}||_{q}=\sqrt[q]{\sum_{i}|\mathbf{a}(i)|^{q}}$;
and $||\mathbf{A}||_{F}$, $||\mathbf{A}||$ and $||\mathbf{a}||$
denote the Frobenius norm, spectrum norm of $\mathbf{A}$ and Euclidean
norm of vector $\mathbf{a}$ respectively.

\section{System Model}

\subsection{C-RAN Topology}

Consider an \emph{uplink} C-RAN system with $M$ distributed single-antenna
RRHs where the RRHs are connected to the BBU pools via the fronthaul,
as illustrated in Figure \ref{fig:Illustration-of-Cloud}. There are
a total of $KN_{c}$ single-antenna UEs being served on $N_{c}$ subcarriers
and each subcarrier is allocated to $K$ UEs. Denote the $k$-th UE
on the $c$-th subcarrier as the $(c,k)$-th UE, and the whole set
of UEs as $\mathcal{U}=\{(c,k):c\in\{1,2,...,N_{c}\},k\in\{1,2,...,K\}\}$.
Denote $H_{ik}^{[c]}\in\mathbb{C}$ as the channel%
\footnote{In $H_{ik}^{[c]}$, $i\in\{1,\cdots M\}$, $c\in\{1,\cdots N_{c}\}$
and $(c,k)\in\mathcal{U}$ denote the RRH index, subcarrier index,
and the UE index respectively. Note that we shall frequently use this
notation rule in the entire paper.%
} from the $(c,k)$-th UE to the $i$-th RRH. The received symbol at
the $i$-th RRH on the $c$-th subcarrier $y_{i}^{[c]}$ can be expressed
as {\small{}
\begin{equation}
y_{i}^{[c]}=\sum_{k=1}^{K}H_{ik}^{[c]}x_{k}^{[c]}+n_{i}^{[c]},\; i\in\{1,\cdots M\},c\in\{1,2,...,N_{c}\},\label{eq:signal_model}
\end{equation}
}where $x_{k}^{[c]}$ is the transmitted signal by the $(c,k)$-th
user and $n_{i}^{[c]}$ is the standard complex Gaussian noise at
the $i$-th RRH on the $c$-th subcarrier.

\begin{figure}
\begin{centering}
\includegraphics[scale=0.6]{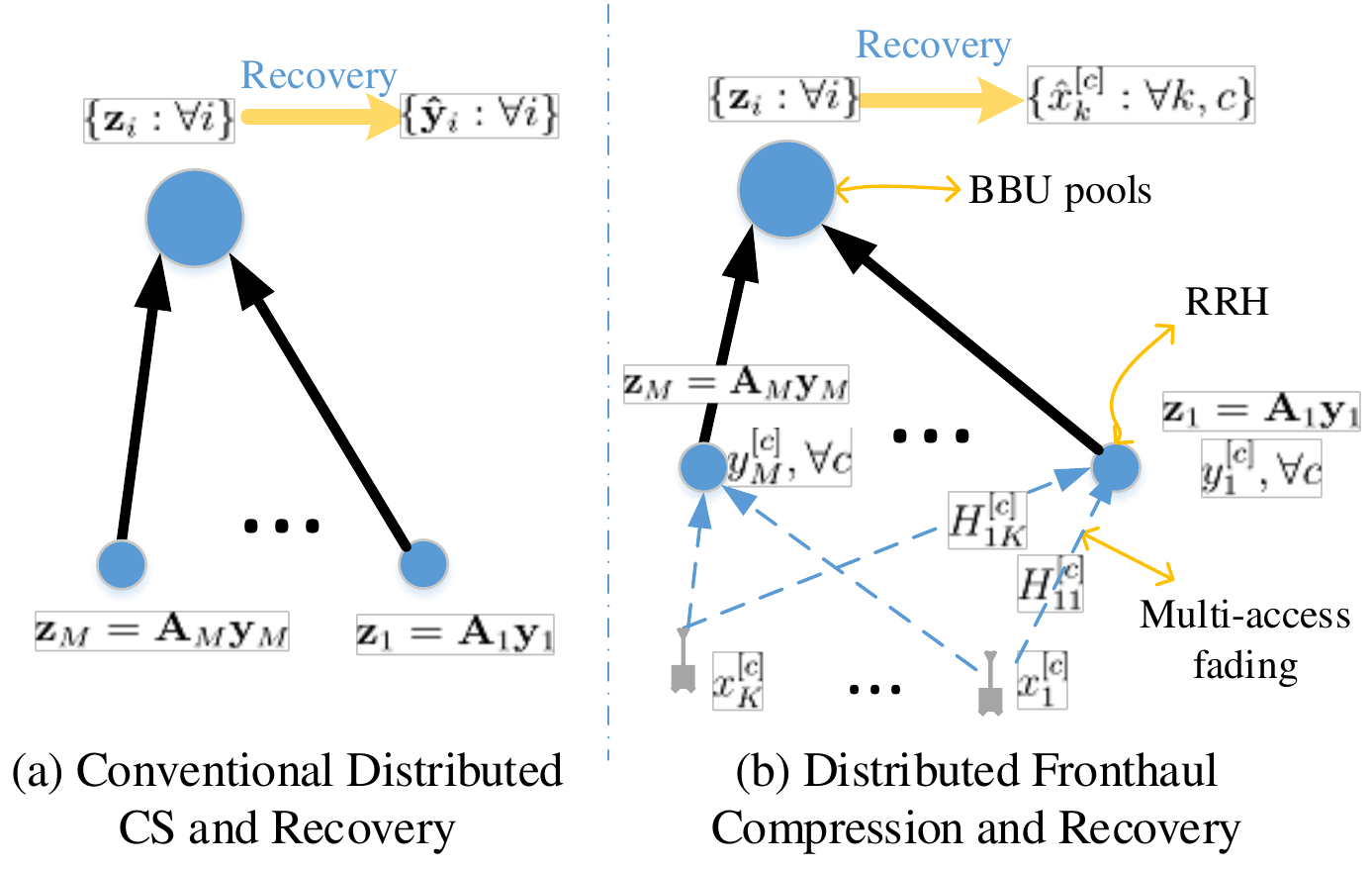}
\par\end{centering}

\protect\caption{\label{fig:Joint-signal-recovery}Illustration of conventional distributed
CS and recovery \cite{baron2005distributed,wang2009distributed} \emph{versus}
the distributed fronthaul compression and signal recovery of C-RAN
systems. Conventional distributed CS aims to recover the vectors $\{\mathbf{y}_{i}\}$
that are compressively and distributively sampled. On the other hand,
in C-RAN systems, the BBU pools aim to jointly recover the \emph{transmitted
signal} $\{x_{k}^{[c]}:\forall k,c\}$ \emph{from the UEs} based on
the distributively compressed data $\mathbf{z}_{i}=\mathbf{A}_{i}\mathbf{y}_{i}$,
$\forall i$ on the fronthaul. }
\end{figure}

Concatenate the received symbols $\{y_{i}^{[c]}:\forall c\}$ and
the noise terms $\{n_{i}^{[c]}:\forall c\}$ at the $i$-th RRH as
$\mathbf{y}_{i}=[\begin{array}{ccc}
y_{i}^{[1]} & \cdots & y_{i}^{[N_{c}]}\end{array}]^{T}\in\mathbb{C}^{N_{c}\times1}$ and $\mathbf{n}_{i}=[\begin{array}{ccc}
n_{i}^{[1]} & \cdots & n_{i}^{[N_{c}]}\end{array}]^{T}\in\mathbb{C}^{N_{c}\times1}$ respectively, and the transmitted symbols $\left\{ x_{k}^{[c]}:\forall(c,k)\in\mathcal{U}\right\} $
from all the UEs as

{\small{}
\begin{equation}
\mathbf{x}=\left[\begin{array}{ccccccc}
x_{1}^{[1]} & \cdots & x_{K}^{[1]}, & \cdots\cdots & ,x_{1}^{[N_{c}]} & \cdots & x_{K}^{[N_{c}]}\end{array}\right]^{T}\in\mathbb{C}^{KN_{c}\times1}.\label{eq:sparse_x}
\end{equation}
}The signal model (\ref{eq:signal_model}) can be re-written as 
\begin{equation}
\mathbf{y}_{i}=\mathbf{H}_{i}\mathbf{x}+\mathbf{n}_{i},\quad i\in\{1,2,\cdots,M\},\label{eq:data_symbol}
\end{equation}
where $\mathbf{H}_{i}$ is the aggregate channel matrix from the the
set of UEs $\mathcal{U}$ to the $i$-th RRH on $N_{c}$ subcarriers
and is given by {\small{}
\begin{equation}
\mathbf{H}_{i}=\left[\begin{array}{ccccccc}
H_{i1}^{[1]} & \cdots & H_{iK}^{[1]} & \mathbf{0} &  & \mathbf{0}\\
 & \mathbf{0} &  & \ddots &  & \mathbf{0}\\
 & \mathbf{0} &  & \mathbf{0} & H_{i1}^{[N_{c}]} & \cdots & H_{iK}^{[N_{c}]}
\end{array}\right]\in\mathbb{C}^{N_{c}\times KN_{c}}.\label{eq:channel_matrix}
\end{equation}
}{\small \par}

We first have the following assumption on the channel model, which
includes both the small scale fading and the large scale path gain
process.
\begin{assumption}
[Channel Model]\label{Channel-ModelThe-channel}The channel $H_{ik}^{[c]}$
is given by $H_{ik}^{[c]}=g_{ik}^{[c]}h_{ik}^{[c]}$, where $g_{ik}^{[c]}\in\mathbb{R}^{+}$
and $h_{ik}^{[c]}\in\mathbb{C}$ denote the large scale fading and
small scale fading parameters, respectively, from the $(c,k)$-th
user to the $i$-th RRH. \end{assumption}
\begin{itemize}
\item \textbf{Small Scale Fading}: The small fading parameters $\{h_{ik}^{[c]}:\forall i,k,c\}$
are i.i.d. standard complex Gaussian variables with zero mean and
unit variance and are independent of $\{g_{ik}^{[c]}:\forall i,k,c\}$. 
\item \textbf{Large Scale Fading}: Let $\mathbf{g}_{k}^{[c]}=\left[\begin{array}{ccc}
\cdots & g_{ik}^{[c]} & \cdots\end{array}\right]_{i\in\{1,...,M\}}$ be the large scale fading parameters from the $(c,k)$-th user to
the set of RRHs $\{1,...,M\}$ and $\mathbf{g}_{k}^{[c]}$ is supposed
to be normalized so that $\left\Vert \mathbf{g}_{k}^{[c]}\right\Vert _{F}^{2}=M$.
Then, $\mathbf{g}_{k}^{[c]}$ and $\mathbf{g}_{k^{'}}^{c'}$ are independent
$\forall(c,k)\neq(c^{'},k^{'})$. Furthermore, $g_{ik}^{[c]}\in[0,\overline{g}]$
is randomly distributed with unit second order moment, i.e., $\mathbb{E}(|g_{ik}^{[c]}|^{2})=1$,
$\forall c,k$. \hfill \QED
\end{itemize}

Note that the effect of normalization of $\mathbf{g}_{k}^{[c]}$ is
incorporated in the different transmit SNR $P_{i}$ for the active
UEs as in (\ref{eq:sparsity_model}). On the other hand, after the
normalization, i.e., $\sum_{i=1}^{M}(g_{ik}^{[c]})^{2}=M$, we still
have the freedom to assume that the random large scale fading parameters
satisfy $\mathbb{E}(|g_{1k}^{[c]}|^{2})=\mathbb{E}(|g_{2k}^{[c]}|^{2})=\cdots=\mathbb{E}(|g_{Mk}^{[c]}|^{2})=1$
from the \emph{uniform} geometric distributions of the users and deployments
of the RRHs. The BBU is assumed to have perfect knowledge of the channel
state information $\{H_{ik}^{[c]}:\forall i,k,c\}$ and this can be
achieved using uplink reference signals%
\footnote{For instance, the UEs send uplink reference signals regularly in TDD-LTE
systems \cite{3GPPphysical}. %
} from the UEs \cite{3GPPphysical}.

From (\ref{eq:data_symbol}), the transmitted symbol vector $\mathbf{x}$
from all the UEs has a full dimension of $KN_{c}$. However, in practice,
the wireless systems might not be fully loaded and $\mathbf{x}$ might
be sparse (e.g., due to the random access mode or the massive bursty
traffic of the UEs \cite{yun2008analysis,wu2012modeling,3gppM2M}).
We illustrate two example scenarios in which the uplink UEs signals
are sparse below. Some other application scenarios that involve sparse
user signals can also be found in \cite{huang2013applications}. 
\begin{itemize}
\item \textbf{Sparse Uplink Signals with Machine Type UEs} \cite{3gppM2M,hasan2013random,Laya2014RACH}:
Suppose the set of $KN_{c}$ machine-type UEs (MTC) \cite{3gppM2M}
access the network independently using random access channel (RACH)
\cite{Laya2014RACH} and the accessing probability of each MTC user
is $p$ ($p\ll1$ for stability \cite{3gppM2M}). Then the average
number of active UEs on the uplink per time slot is $pKN_{c}\ll KN_{c}$,
and hence the uplink signals tend to be sparse. 
\item \textbf{Sparse Uplink Signals with Bursty Applications} \cite{wu2012modeling,yun2008analysis}:
Suppose the set of $KN_{c}$ UEs is running delay-sensitive applications
\cite{wu2012modeling,yun2008analysis} and the inter-arrival time
between two data packets from each user follows from an exponential
distribution with mean of $\lambda$ ($\lambda>1$) time slots independently.
Then the average number of active UEs in the considered time slot
is $\frac{1}{\lambda}KN_{c}<KN_{c}$, and hence the uplink signals
tend to be sparse (e.g., Figure 2 in \cite{yun2008analysis}). 
\end{itemize}

Note that bursty uplink traffic cannot be scheduled effectively in
cellular systems due to the latency%
\footnote{To schedule uplink transmissions \cite{3gppM2M}, the UE first sends
a BW request through random access and the Node B then schedules the
uplink resource. As such, uplink scheduling has significant protocol
overhead and latency and it is not suitable to support massive UEs
with small burst, low latency and long duty cycle requirement. %
} concerns. For instance, in LTE-A systems \cite{3gppM2M}, the random
access is commonly advocated \cite{3gppM2M,hasan2013random,Laya2014RACH}
as an efficient protocol to accommodate a large number of machine-type
communications with massive burst of low latency data. Denote the
set of non-zero elements of $\mathbf{x}$ as $\mathcal{T}$, i.e.,
\begin{equation}
\begin{cases}
\mathbf{x}(i)\sim\mathcal{CN}(0,P_{i}), & i\in\mathcal{T}\\
\mathbf{x}(i)=0, & i\notin\mathcal{T}
\end{cases},\label{eq:sparsity_model}
\end{equation}
where $P_{i}$ is the transmit SNR of the user corresponding to $\mathbf{x}(i)$,
$i\in\mathcal{T}$. Let $s\triangleq\#\left|\mathcal{T}\right|\ll KN_{c}$
and there are $s$ \emph{active}%
\footnote{Here we call the $(c,k)$-th UE \emph{active} if it is transmitting
a signal, i.e., $x_{k}^{[c]}$ in (\ref{eq:sparse_x}) is non-zero. %
} UEs among $\mathcal{U}$ in the considered time slot. Our target
is to exploit the underlying sparsity in the uplink transmitted signals
to reduce the fronthaul loading in C-RAN systems. The proposed scheme
consists of \emph{distributed compression} at the RRH as well as efficient
\emph{joint recovery }at the BBU pools. For robust implementation,
the BBU pool in the C-RAN does not know the signal sparsity support
$\mathcal{T}$ (the set of active UEs) and the sparsity level $s\triangleq\#\left|\mathcal{T}\right|$
(the number of active UEs). We shall elaborate the distributed fronthaul
compression scheme in the next section. \vspace{-0.4cm}

\subsection{Distributed Fronthaul Compression in C-RAN Systems}

In this section, we shall propose a distributed fronthaul compression
scheme in C-RAN systems. As we can see in (\ref{eq:data_symbol}),
the RRH can compress the received signals from the multi-access fading
channel $\{\mathbf{y}_{i}:\forall i\}$ before sending them to the
BBU pools for the joint recovery of $\mathbf{x}$. Suppose the $i$-th
RRH uses a \emph{local} compression matrix $\mathbf{A}_{i}\in\mathbb{C}^{R\times N_{c}}$
to compress the $(N_{c}\times1)$ symbol vector $\mathbf{y}_{i}$
into a lower dimensional $(R\times1)$ signal vector $\mathbf{z}_{i}$
($R\leq N_{c}$) by 
\begin{equation}
\mathbf{z}_{i}=\mathbf{A}_{i}\mathbf{y}_{i}.\label{eq:compression}
\end{equation}
Then the total number of measurements in $\{\mathbf{z}_{i}:\forall i\}$
to be transmitted to the BBU pools on the fronthaul is $MR$ instead
of $MN_{c}$. This represents a compression rate of $\alpha=\dfrac{R}{N_{c}}\leq1$.
Note that when $\alpha=1$ (i.e., $R=N_{c}$), (\ref{eq:compression})
is reduced to the scenario with no compression. On the other hand,
the compression in (\ref{eq:compression}) is a simple linear operation
and hence it can be easily incorporated in the RRH of C-RAN systems. 
\begin{remrk}
[Consideration of Quantization in Fronthaul]In this paper, the fronthaul
compression is achieved by reducing the dimensions of the signal $\mathbf{z}_{i}$
to be sent from the RRH to the BBU cloud. In practice, the absolute
fronthaul loading (b/s) depends on both the \emph{dimensions of} $\mathbf{z}_{i}$
(i.e., $R$) and\emph{ the quantization}. The \emph{dimension reduction
of }$\mathbf{z}_{i}$\emph{ at the RRH} plays a first order role in
fronthaul compression because a smaller signal dimension means a smaller
number of complex numbers need to be quantized from the RRH to the
BBU. For instance, from classical quantization theory \cite{zador1982asymptotic,dai2011information},
to keep a constant distortion $\Delta\triangleq\frac{\left\Vert \mathbf{z}_{i}-\mathbf{\hat{z}}_{i}\right\Vert ^{2}}{\left\Vert \mathbf{z}_{i}\right\Vert ^{2}}$
(where $\mathbf{\hat{z}}_{i}$ denotes the quantized vector of $\mathbf{z}_{i}$),
the required number of bits for quantization should scale linearly
with the signal dimension $R$ and is given by $B=\mathcal{O}(R\log\frac{1}{\Delta})$.
As such, there is a genuine compression in (\ref{eq:compression})
even when quantization effect is included. 
\end{remrk}

Before we describe the distributed fronthaul compression scheme, we
first give the generation method for the local compression matrix
$\{\mathbf{A}_{i}\}$.
\begin{definitn}
[Local Compression Matrices $\mathbf{A}_{i}$]\label{Local-Compression-Matrices}The
$(r,c)$-th entry of the $R\times N_{c}$ matrix $\mathbf{A}_{i}$
is given by $\mathbf{A}_{i}(r,c)=\sqrt{\frac{1}{MR}}\exp(j\theta_{irc})$,
where $\theta_{irc}$ is i.i.d. drawn from the uniform distribution
over $[0,2\pi)$. 
\end{definitn}

\begin{remrk}
[Interpretation of Definition \ref{Local-Compression-Matrices}]Note
that $\mathbf{A}_{i}$ serves as the local \emph{measurement matrix}
of $\mathbf{y}_{i}$, as in (\ref{eq:compression}). In the CS literature,
it is shown that efficient and robust CS recovery can be achieved
when the measurement matrix satisfies a proper RIP condition \cite{candes2008introduction},
and measurement matrices randomly generated from sub-Gaussian distribution
\cite{candes2008introduction} can satisfy the RIP with overwhelming%
\footnote{For an $R\times N$ matrix $\Phi$ i.i.d. generated from random sub-Gaussian
distributions, it is shown that when $R=c_{1}s\log N$, the probability
that $\Phi$ fails to satisfy the $s$-th RIP with $\delta$ decays
exponentially w.r.t. $R$ as $O(\textrm{exp}(-c_{2}R)$, where $c_{1}$
and $c_{2}$ are positive constants depending on $\delta$ \cite{candes2008introduction}.%
} probability \cite{candes2008introduction}. As such, this randomized
generation method has been widely used in the literature and we adopt
this conventional approach to generate the local compression matrix
$\{\mathbf{A}_{i}\}$ (as in Definition \ref{Local-Compression-Matrices})
with a good RIP. The factor $\sqrt{\frac{1}{MR}}$ in each entry in
$\mathbf{A}_{i}$ is for normalization. On the other hand, the compression
matrices $\{\mathbf{A}_{i}\}$ are generated offline so that it is
available at the BBU%
\footnote{For instance, the compression matrices $\{\mathbf{A}_{i}\}$ can be
initialized from the BBU and the distributed to the RRH during the
setup of the C-RAN so that the BBU has the knowledge of $\{\mathbf{A}_{i}\}$%
}. 
\end{remrk}

The overall distributed fronthaul compression scheme in the C-RAN
system can be described as follows (as illustrated in Figure \ref{fig:Joint-signal-recovery}):

\emph{Algorithm 1 (Distributed Fronthaul Compression)}
\begin{itemize}
\item \textbf{Step 1} \emph{(Reception at RRHs)}: The $i$-th RRH receives
the channel outputs $\mathbf{y}_{i}=\mathbf{H}_{i}\mathbf{x}+\mathbf{n}_{i}$
as in (\ref{eq:data_symbol}).
\item \textbf{Step 2} \emph{(Distributed Compression on the Fronthaul)}:
The $i$-th RRH compresses the received $N_{c}$-dimensional data
vector $\mathbf{y}_{i}$ into an $R$-dimensional vector $\mathbf{z}_{i}$
with a local compression matrix $\mathbf{A}_{i}\in\mathbb{C}^{R\times N_{c}}$
by $\mathbf{z}_{i}=\mathbf{A}_{i}\mathbf{y}_{i}$, as in (\ref{eq:compression}).
The $i$-th RRH sends the compressed vector $\mathbf{z}_{i}$ to the
BBU pools on the fronthaul. 
\item \textbf{Step 3} \emph{(Centralized Signal Recovery at BBU Pools)}:
The BBU pools collect the compressed data symbols from all RRHs on
the fronthaul, i.e., $\{\mathbf{z}_{i}:i=1,2,...,M\}$ and then jointly
recover the transmitted signal $\{x_{i}^{[c]}:\forall i,c\}$ from
the UEs. \hfill \QED
\end{itemize}

Note that in Algorithm 1, the remaining question is how to conduct
the signal recovery at the BBUs (Step 3) from the measurements $\{\mathbf{z}_{i}\}$
with reduced dimension (i.e., $MR\leq MN_{c}$). Different from classical
distributed CS \cite{wang2009distributed,baron2005distributed}, in
which the objective is to recover the signal vectors that are compressed
distributively (as illustrated in Figure \ref{fig:Joint-signal-recovery}(a)),
our goal is to recover the transmitted signal $\mathbf{x}$ from the
UEs (as in Figure \ref{fig:Joint-signal-recovery}(b)). Existing results
on the joint recovery of distributed CS \cite{baron2005distributed,wang2009distributed}
cannot be applied, and we need to extend these results to incorporate
the underlying multiaccess fading channels $\{H_{ik}^{[c]}:\forall i,k,c\}$
in C-RAN systems. \vspace{-0.5cm}

\begin{center}
\fbox{\begin{minipage}[t]{1\columnwidth}%
Challenge 1: End-to-end signal recovery of $\mathbf{x}$ at the BBU
from distributed CS measurements at the RRHs.%
\end{minipage}}
\par\end{center}

\section{CS-enabled Signal Recovery in C-RAN Systems}

In this section, we first elaborate the proposed signal recovery (i.e.,
step 3 in Algorithm 1), which exploits the signal sparsity and the
structure of the multi-access fading channel to recover the transmitted
signals $\{x_{i}^{[c]}:\forall i,c\}$ by the $KN_{c}$ UEs. Based
on that, we establish sufficient conditions on the required number
of measurements at RRHs to achieve correct active user detection \emph{with
high probability} in the proposed algorithm. Note that based on the
performance result of the correct active user detection, we shall
further quantify the C-RAN capacity in Section IV.\vspace{-0.5cm}

\subsection{End-to-End Sparse Signal Recovery Algorithm}

The proposed signal recovery algorithm at the BBU pools consists of
three major components: i) rough signal estimation of $\hat{\mathbf{x}}$
using CS techniques; ii) active user detection $\hat{\mathcal{T}}$
based on $\hat{\mathbf{x}}$; and iii) zero-forcing (ZF) receiver
based on $\hat{\mathcal{T}}$. Figure \ref{fig:Illustration-of-Algorithm2}
illustrates the block diagram of the signal recovery algorithm at
the BBU pools. Note that the proposed recovery algorithm is different
from the standard CS recovery in \cite{candes2005decoding,tropp2007signal,needell2009cosamp}
because the proposed recovery algorithm has embraced the multi-access
fading channel $\{H_{ik}^{[c]}:\forall i,k,c\}$ in the recovery process.
The three components are elaborated in detail below. 
\begin{itemize}
\item \textbf{Rough signal estimation of }$\hat{\mathbf{x}}$\textbf{ using
CS techniques}: We first formulate the rough signal estimation of
$\hat{\mathbf{x}}$ into a standard CS problem. First, concatenate
the compressed symbol vectors $\{\mathbf{z}_{i}\in\mathbb{C}^{R\times1}:i=1,2,...,M\}$
into a long vector to be $\mathbf{z}=\left[\begin{array}{ccc}
\mathbf{z}_{1}^{T} & \cdots & \mathbf{z}_{M}^{T}\end{array}\right]^{T}\in\mathbb{C}^{MR\times1}$. Then, (\ref{eq:data_symbol}) can be equivalently written as 
\begin{equation}
\mathbf{z}=\Theta\mathbf{x}+\mathbf{n},\label{eq:standard_model}
\end{equation}
 where $\mathbf{n}$ is the aggregate measurement noise 
\begin{equation}
\mathbf{n}=\left[\begin{array}{ccc}
(\mathbf{A}_{1}\mathbf{n}_{1})^{T} & \cdots & (\mathbf{A}_{M}\mathbf{n}_{M})^{T}\end{array}\right]^{T}\in\mathbb{C}^{MR\times1},\label{eq:aggregate_noise}
\end{equation}
and $\Theta$ is given by (a detailed expression of $\Theta$ is also
given in (\ref{eq:expansion}))
\begin{equation}
\Theta=\left[\begin{array}{ccc}
(\mathbf{A}_{1}\mathbf{H}_{1})^{T} & \cdots & (\mathbf{A}_{M}\mathbf{H}_{M})^{T}\end{array}\right]^{T}\in\mathbb{C}^{MR\times KN_{c}}.\label{eq:theta}
\end{equation}
Therefore, (\ref{eq:standard_model}) matches the standard CS model
in which $\mathbf{z}$ is the measurements, $\Theta$ is the aggregate
\emph{measurement matrix}, $\mathbf{x}$ is the sparse signal vector
to be recovered and $\mathbf{n}$ is the CS measurement noise (with
$\mathbb{E}\left(||\mathbf{n}||^{2}\right)=N_{c}$). Note that knowledge
of $\Theta$ can be available at the BBU from (\ref{eq:theta}) and
the fact that both $\{\mathbf{H}_{i}\}$ and $\{\mathbf{A}_{i}\}$
are available at the BBU. Since the above formulation (\ref{eq:standard_model})
has incorporated the channel fading of $\{H_{ik}^{[c]}\}$ into the
CS measurement model, it enables us to conduct the end-to-end recovery
of $\mathbf{x}$. Using the classical basis%
\footnote{Among various CS recovery algorithms \cite{candes2005decoding,tropp2007signal,needell2009cosamp,dai2009subspace}
(including BP \cite{candes2005decoding}, OMP \cite{tropp2007signal},
CoSaMP \cite{needell2009cosamp} and SP \cite{dai2009subspace}),
the BP is shown to has competitive theoretical as well as empirical
recovery performance \cite{candes2005decoding}, and it does not require
the knowledge of the sparsity level (i.e., $\#|\mathcal{T}|$). Therefore,
in this paper, we adopt the BP approach to conduct the rough estimation
of $\hat{\mathbf{x}}$ in Step 1. %
} pursuit (BP) \cite{candes2005decoding} CS recovery technique, the
rough signal recovery of $\hat{\mathbf{x}}$ can be formulated as
\begin{eqnarray}
\mathcal{P}_{1}:\qquad\min_{\hat{\mathbf{x}}} &  & \left\Vert \hat{\mathbf{x}}\right\Vert _{1}\nonumber \\
\mbox{s.t.} &  & \left\Vert \Theta\hat{\mathbf{x}}-\mathbf{z}\right\Vert \leq\lambda.\label{eq:orginal_form-1}
\end{eqnarray}
where $\lambda$ is a proper threshold parameter. 
\item \textbf{Active user detection $\hat{\mathcal{T}}$ based on} $\hat{\mathbf{x}}$:
The following criterion is proposed to detect the set of active UEs
$\hat{\mathcal{T}}$: {\small{}
\begin{eqnarray}
\hat{\mathcal{T}} & = & \arg\min_{\hat{\mathcal{T}}}\#|\hat{\mathcal{T}}|,\label{eq:active_user_detection}\\
 &  & \mbox{s.t.}\begin{cases}
|\mathbf{\hat{x}}(i)|\geq|\mathbf{\hat{x}}(j)|,\;\forall i\in\hat{\mathcal{T}},\quad j\notin\mathcal{T},\\
\left\Vert \left(\mathbf{I}-\Theta_{\hat{\mathcal{T}}}\Theta_{\hat{\mathcal{T}}}^{\dagger}\right)\mathbf{z}\right\Vert \leq\lambda\;\mbox{OR}\;\#|\hat{\mathcal{T}}|=MR.
\end{cases}\nonumber 
\end{eqnarray}
}where $\Theta_{\hat{\mathcal{T}}}^{\dagger}\triangleq(\Theta_{\hat{\mathcal{T}}}^{H}\Theta_{\hat{\mathcal{T}}})^{-1}\Theta_{\hat{\mathcal{T}}}^{H}$
denotes the Moore-Penrose pseudoinverse of $\Theta_{\mathcal{\hat{T}}}$.
\item \textbf{Zero-forcing (ZF) receiver based on} $\hat{\mathcal{T}}$:
The BBU pools apply the zero-forcing receiver $\mathbf{U}=(\Theta_{\mathcal{\hat{T}}}^{\dagger})^{H}$
\cite{jiang2011performance} to recover the transmitted signals from
the UEs in $\hat{\mathcal{T}}$, i.e., $\hat{\mathbf{x}}_{\hat{\mathcal{T}}}=\Theta_{\mathcal{T}}^{\dagger}\mathbf{z}$,
$\hat{\mathbf{x}}_{\{1,..,KN_{c}\}\backslash\hat{\mathcal{T}}}=\mathbf{0}$. 
\end{itemize}

The overall signal recovery at the BBU pools is summarized in Algorithm
2 (as well as in Figure \ref{fig:Illustration-of-Algorithm2}).

\begin{figure}
\begin{centering}
\includegraphics{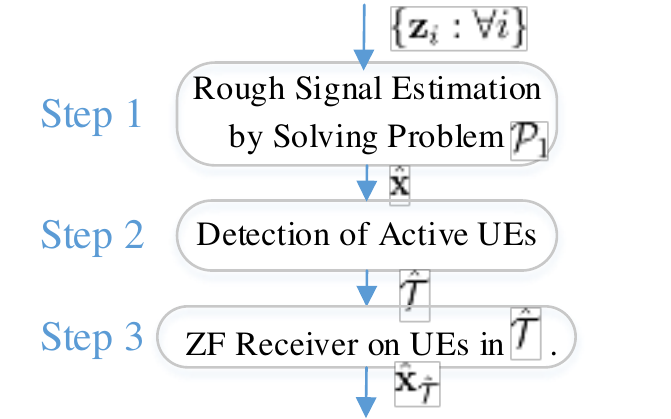}
\par\end{centering}

\protect\caption{\label{fig:Illustration-of-Algorithm2}Illustration of the three major
steps in signal recovery at the BBU pools (Algorithm 2).\vspace{-0.5cm}
}

\end{figure}

\emph{Algorithm 2 (Joint Signal Recovery in Uplink C-RAN Systems with
CS Techniques)}
\begin{itemize}
\item \textbf{Step 1 }\emph{(Rough Signal Estimation Using BP)}: Solve Problem
$\mathcal{P}_{1}$ to obtain $\hat{\mathbf{x}}$.
\item \textbf{Step 2 }\emph{(Active User Detection)}: Find the solution
$\hat{\mathcal{T}}$ to (\ref{eq:active_user_detection}) via the
following steps: First sort $\{|\mathbf{\hat{x}}(i)|:i=1,..,KN_{c}\}$
in descending order to obtain the sorted indices $\{i_{1},i_{2},..,i_{KN_{c}}\}$,
initialize $k=1$ and $\hat{\mathcal{T}}=\{i_{1}\}$ and then execute
the following to find $\hat{\mathcal{T}}$. 

\begin{itemize}
\item \emph{(Greedy Detection)}:If $\left\Vert \left(\mathbf{I}-\Theta_{\hat{\mathcal{T}}}\Theta_{\hat{\mathcal{T}}}^{\dagger}\right)\mathbf{z}\right\Vert \leq\lambda$
or%
\footnote{Note that we have constrained the maximum size of $\hat{\mathcal{T}}$
to be $MR$ as in (\ref{eq:active_user_detection}). This is because
we have a total of $MR$ observations at the BBU so that we can recover
$MR$ unknowns at most. %
} $k=MR$, then stop and output $\hat{\mathcal{T}}$. Else, update
$k=k+1$, $\hat{\mathcal{T}}=\hat{\mathcal{T}}\bigcup\{i_{k}\}$ and
repeat. 
\end{itemize}
\item \textbf{Step 3} \emph{(Zero-forcing Receiver)}: Based on the estimated
set of UEs $\hat{\mathcal{T}}$, apply the linear zero-forcing receiver
$\mathbf{U}=(\Theta_{\mathcal{\hat{T}}}^{\dagger})^{H}$ \cite{jiang2011performance}
to recover the transmitted signals from the UEs in $\hat{\mathcal{T}}$,
i.e., $\hat{\mathbf{x}}_{\hat{\mathcal{T}}}=(\Theta_{\mathcal{\hat{T}}}^{\dagger})\mathbf{z}$,
$\hat{\mathbf{x}}_{\{1,..,KN_{c}\}\backslash\hat{\mathcal{T}}}=\mathbf{0}$. 
\end{itemize}

\begin{remrk}
[Implementation Consideration of Algorithm 2]Note that the computation
of Algorithm 2 is dominated by solving the $l_{1}$-norm minimization
(Problem $\mathcal{P}_{1}$), which can be accomplished with a complexity
of $O\left(\left(KN_{c}\right)^{3}\right)$ \cite{candes2008introduction}.
On the other hand, Algorithm 2 does not require knowledge of the sparsity
level $\#|\mathcal{T}|$ (i.e., the number of active UEs) because
this can be automatically detected in Step 2 of Algorithm 2. Finally,
in Step 3 of Algorithm 2, the classical ZF receiver is adopted because
of its simplicity and asymptotical optimality in high SNR regions
\cite{jiang2011performance}. \hfill \QED
\end{remrk}

\begin{figure*}
\begin{centering}
\vspace*{-0.5cm}

{\begin{equation}{} \Theta=\left[\begin{array}{ccc} \left[\begin{array}{ccc} H_{11}^{[1]}\mathbf{A}_{1}(1,1) & \cdots & H_{1K}^{[1]}\mathbf{A}_{1}(1,1)\\ \vdots & \ddots & \vdots\\ H_{11}^{[1]}\mathbf{A}_{1}(R,1) & \cdots & H_{1K}^{[1]}\mathbf{A}_{1}(R,1) \end{array}\right] & \cdots & \left[\begin{array}{ccc} H_{11}^{[N_{c}]}\mathbf{A}_{1}(1,N_{c}) & \cdots & H_{1K}^{[N_{c}]}\mathbf{A}_{1}(1,N_{c})\\ \vdots & \ddots & \vdots\\ H_{11}^{[N_{c}]}\mathbf{A}_{1}(R,N_{c}) & \cdots & H_{1K}^{[N_{c}]}\mathbf{A}_{1}(R,N_{c}) \end{array}\right]\\ \vdots & \ddots & \vdots\\ \left[\begin{array}{ccc} H_{M1}^{[1]}\mathbf{A}_{M}(1,1) & \cdots & H_{MK}^{[1]}\mathbf{A}_{M}(1,1)\\ \vdots & \ddots & \vdots\\ H_{M1}^{[1]}\mathbf{A}_{M}(R,1) & \cdots & H_{MK}^{[1]}\mathbf{A}_{M}(R,1) \end{array}\right] & \cdots & \left[\begin{array}{ccc} H_{M1}^{[N_{c}]}\mathbf{A}_{M}(1,N_{c}) & \cdots & H_{MK}^{[N_{c}]}\mathbf{A}_{M}(1,N_{c})\\ \vdots & \ddots & \vdots\\ H_{M1}^{[N_{c}]}\mathbf{A}_{M}(R,N_{c}) & \cdots & H_{MK}^{[N_{c}]}\mathbf{A}_{M}(R,N_{c}) \end{array}\right] \end{array}\right]\label{eq:expansion} \end{equation} }
\hrulefill
\vspace*{-0.5cm}
\end{centering}
\end{figure*}

\subsection{Analysis of Correct Active User Detection in Algorithm 2}

Note that when the detected set of active UEs in Algorithm 2 is correct,
i.e., $\hat{\mathcal{T}}=\mathcal{T}$, Step 3 in Algorithm 2\emph{
}will be\emph{ reduced to classical ZF} on the set of the active uplink
UEs $\mathcal{T}$ directly. This is the desired performance and it
is also asymptotically optimal because of the asymptotical optimality
of ZF in high SNR \cite{jiang2011performance}. Under correct active
user detection, further analysis \cite{jiang2011performance} on ZF
can also be conducted to give further characterization of the C-RAN
communication performance (e.g., sum capacity, which will be discussed
in Section IV). Therefore, it is critical to first understand the
performance of active user detection in Algorithm 2.\vspace{-0.5cm}

\begin{center}
\fbox{\begin{minipage}[t]{1\columnwidth}%
Challenge 2: Analyze the probability of correct active user detection,
i.e., $\hat{\mathcal{T}}=\mathcal{T}$ in Algorithm 2. %
\end{minipage}}
\par\end{center}

In Algorithm 2, the rough estimation of $\hat{\mathbf{x}}$ in \textbf{Step
1} plays an important role in the detection of active user $\hat{\mathcal{T}}$
in \textbf{Step 2}. As such, we shall first characterize Step 1 in
Algorithm 2 (the BP in Problem $\mathcal{P}_{1}$). In the literature,
RIP \cite{candes2005decoding} is commonly adopted to facilitate the
performance analysis of the BP. We first review the notion of the
RIP in Definition \ref{RIPA-matrix} and the associated performance
result of BP under RIP in Lemma \ref{Robust-Recovery-under}.
\begin{definitn}
[Restricted Isometry Property \cite{candes2005decoding}]\label{RIPA-matrix}Matrix
$\Theta\in\mathbb{C}^{MR\times KN_{c}}$ satisfies a $k$-th order
RIP with a prescribed\emph{ restricted isometry constant} (RIC) $\delta$
if $0\leq\delta<1$ and 
\[
(1-\delta)\left\Vert \mathbf{x}\right\Vert ^{2}\leq\left\Vert \Theta\mathbf{x}\right\Vert ^{2}\leq(1+\delta)\left\Vert \mathbf{x}\right\Vert ^{2}
\]
holds for all $\mathbf{x}\in\mathbb{C}^{KN_{c}\times1}$ where $\left\Vert \mathbf{x}\right\Vert _{0}\leq k$.
\hfill \QED
\end{definitn}

Denote event $\mathcal{\mathcal{E}}_{k,\delta}$ as follows:
\begin{equation}
\mathcal{\mathcal{E}}_{k,\delta}:\qquad\Theta\mbox{ satisfies a \ensuremath{k}-th order RIP with RIC \ensuremath{\delta}}.\label{eq:event_rip}
\end{equation}
From \cite{candes2005decoding}, we have the following lemma on the
recovery performance regarding problem $\mathcal{P}_{1}$. 
\begin{lemma}
[Robust Recovery of BP under RIP \cite{candes2008introduction}]\label{Robust-Recovery-under}Suppose
the threshold parameter $\lambda$ in $\mathcal{P}_{1}$ satisfies
$\lambda\geq\left\Vert \mathbf{n}\right\Vert _{F}$. If event $\mathcal{\mathcal{E}}_{2s,\delta}$
happens, where $s\triangleq\#\left|\mathcal{T}\right|\ll KN_{c}$,
and the RIC $\delta$ satisfies $\delta<\sqrt{2}-1$, then the solution
$\hat{\mathbf{x}}$ to $\mathcal{P}_{1}$ satisfies 
\begin{equation}
\left\Vert \mathbf{x}-\hat{\mathbf{x}}\right\Vert \leq c_{2}\lambda,\label{eq:error_bound}
\end{equation}
where $c_{1}=\frac{\lambda^{2}-N_{c}}{4N_{c}}$ and $c_{2}=4\frac{\sqrt{1+\delta}}{1-(1+\sqrt{2})\delta}$.
\hfill \QED
\end{lemma}

Based on Lemma \ref{Robust-Recovery-under}, we obtain the following\emph{
conditions} for correct active user detection and the associated performance
result in Algorithm 2. 
\begin{thm}
[Correct Active User Detection Conditions]\label{auxiliary_lemma}Suppose
$\lambda\geq||\mathbf{n}||_{F}$ and event $\mathcal{\mathcal{E}}_{2s,\delta}$
happens where $\delta<\sqrt{2}-1$. If event $\mathcal{\mathcal{E}}_{1}$
is true where {\small{}
\[
\mathcal{\mathcal{E}}_{1}:\;|\mathbf{x}(i)|>\sqrt{2}c_{2}\lambda,\;\forall i\in\mathcal{T},
\]
} $c_{2}$ is in Lemma \ref{Robust-Recovery-under}, then 

1) the detected $\hat{\mathcal{T}}$ in Step 2 of Algorithm 2 is correct,
i.e., $\hat{\mathcal{T}}=\mathcal{T}$.

2) the recovered signal $\hat{\mathbf{x}}$ from Step 3 of Algorithm
2 satisfies $\left\Vert \mathbf{x}-\hat{\mathbf{x}}\right\Vert \leq\frac{1}{\sqrt{1-\delta}}\left\Vert \mathbf{n}\right\Vert $.\end{thm}
\begin{proof}
See Appendix \ref{sub:Proof-of-auxililary-lemma}.
\end{proof}

\begin{remrk}
[Fronthaul Quantization]The results of Theorem \ref{auxiliary_lemma}
also cover the case with quantization in the fronthaul as well. Suppose
that the signal $\mathbf{z}$ is further \emph{quantized} with $b$
quantization bits per dimension. The quantized $\hat{\mathbf{z}}$
is given by $\hat{\mathbf{z}}=\mathbf{z}+\hat{\mathbf{n}}$, where
$\hat{\mathbf{n}}$ is the quantization noise ($||\hat{\mathbf{n}}||$
scales in the order of $\mathcal{O}(2^{-\frac{b}{2}})||\mathbf{z}||\sim\mathcal{O}\left(2^{-\frac{b}{2}}\sqrt{\sum_{i\in\mathcal{T}}P_{i}}\right)$
\cite{zador1982asymptotic,dai2011information}). Then the aggregate
signal model becomes $\hat{\mathbf{z}}=\mathbf{z}+\hat{\mathbf{n}}=\Theta\mathbf{x}+\mathbf{n}+\hat{\mathbf{n}}$,
where $(\mathbf{n}+\hat{\mathbf{n}})$ replaces the role of $\mathbf{n}$
in (\ref{eq:standard_model}) with $\left\Vert \hat{\mathbf{n}}+\mathbf{n}\right\Vert \leq\left\Vert \hat{\mathbf{n}}\right\Vert +\left\Vert \mathbf{n}\right\Vert \sim\mathcal{O}\left(2^{-\frac{b}{2}}\sqrt{\sum_{i\in\mathcal{T}}P_{i}}\right)+\left\Vert \mathbf{n}\right\Vert $.
From Theorem \ref{auxiliary_lemma}, $\hat{\mathcal{T}}=\mathcal{T}$
can still be achieved in Step 2 and the final recovered signal $\hat{\mathbf{x}}$
in Step 3 satisfies $\left\Vert \mathbf{x}-\hat{\mathbf{x}}\right\Vert \leq\frac{1}{\sqrt{1-\delta}}\left\Vert \hat{\mathbf{n}}+\mathbf{n}\right\Vert $
in Algorithm 2. 
\end{remrk}

Note that classical BP CS recovery (e.g., Problem $\mathcal{P}_{1}$,)
requires that the measurement noise is bounded (e.g., $||\mathbf{n}||_{F}\leq\lambda$
as in Lemma 1) \cite{candes2008introduction}. However, due to the
Gaussian factors in the noise (i.e., $\{\mathbf{n}_{i}:\forall i\}$
in (\ref{eq:aggregate_noise})), the aggregate noise $||\mathbf{n}||_{F}$
may not always be bounded. By using the tools from concentration inequalities
\cite{boucheron2004concentration}, we first show below that $||\mathbf{n}||_{F}$
can be bounded \emph{with high probability,} despite the complicated
form of the colored noise $\mathbf{n}$ in (\ref{eq:aggregate_noise}).
\begin{lemma}
[Bounded Noise with High Probability]\label{Probability-of-Bounded-noise}Suppose
$\lambda\geq\sqrt{2N_{c}}$. The probability that $||\mathbf{n}||_{F}\leq\lambda$
happens is at least $1-\exp\left(-c_{1}M\right)$, where $c_{1}=\frac{\lambda^{2}-N_{c}}{4N_{c}}$.\end{lemma}
\begin{proof}
See Appendix \ref{sub:Proof-of-Lemma-Bounded}.
\end{proof}

From Theorem \ref{auxiliary_lemma} and Lemma \ref{Probability-of-Bounded-noise},
we further obtain the following lower bound on the probability of
correct active user detection.
\begin{thm}
[Probability of Correct Active User Detection]\label{Probability-of-Correct}The
probability of correct active UEs $\hat{\mathcal{T}}$ detection,
i.e., $\hat{\mathcal{T}}=\mathcal{T}$, in Step 2 of Algorithm 2 satisfies
\begin{eqnarray}
 &  & \Pr\left(\hat{\mathcal{T}}=\mathcal{T}\right)\geq\Pr\left(\mathcal{\mathcal{E}}_{2s,\delta}\right)\times\label{eq:Bound}\\
 &  & \left(1-\exp\left(-c_{1}M\right)-s\cdot\left(1-\exp\left(-\frac{2\left(c_{2}\lambda\right)^{2}}{P_{\min}}\right)\right)\right)\nonumber 
\end{eqnarray}
where $P_{\min}\triangleq\min_{i\in\mathcal{T}}P_{i}$, $c_{1}=\frac{\lambda^{2}-N_{c}}{4N_{c}}$,
$c_{2}$ is in Lemma \ref{Robust-Recovery-under} and $\Pr\left(\mathcal{\mathcal{E}}_{2s,\delta}\right)$
denotes the probability that event $\mathcal{\mathcal{E}}_{k,\delta}$
in (\ref{eq:event_rip}) happens with parameters $k=2s$ and $\delta<\sqrt{2}-1$.\end{thm}
\begin{proof}
See Appendix \ref{sub:Proof-of-Lemma-active_bound}.
\end{proof}

Note that the result of correct active user detection in Theorem \ref{Probability-of-Correct}
is important for further characterization of the uplink C-RAN capacity
(Theorem \ref{Bound-of-Average} in Section IV). From (\ref{eq:Bound}),
in high SNR regimes, i.e., $P_{\min}\triangleq\min_{i\in\mathcal{T}}P_{i}\gg4c_{2}^{2}N_{c}$,
with $\lambda$ in Problem $\mathcal{P}_{1}$ being adaptive to $P_{\min}$
(e.g., $\lambda=\left(\frac{P_{\min}N_{c}}{4c_{2}^{2}}\right)^{\frac{1}{4}}$),
the factor after $\Pr\left(\mathcal{\mathcal{E}}_{2s,\delta}\right)$
in (\ref{eq:Bound}) approaches 1, and hence (\ref{eq:Bound}) is
simplified to be $\Pr\left(\hat{\mathcal{T}}=\mathcal{T}\right)\gtrsim\Pr\left(\mathcal{\mathcal{E}}_{2s,\delta}\right)$.
Consequently, the probability of $\hat{\mathcal{T}}=\mathcal{T}$
can be lower bounded by the probability of event $\mathcal{\mathcal{E}}_{2s,\delta}$
only. This relationship under high SNR simplifies the analysis and
potentially leads to elegant results. Note that in (\ref{eq:Bound}),
it remains unknown what $\Pr\left(\mathcal{\mathcal{E}}_{2s,\delta}\right)$
is. A more fundamental question is whether $\mathcal{\mathcal{E}}_{2s,\delta}$
can happen or not (i.e., whether the RIP condition can hold for $\Theta$).
We shall investigate this issue in the next section.

\subsection{Characterization of the RIP for $\Theta$}

In this section, we justify that the RIP condition can be satisfied
for our aggregate measurement matrix $\Theta$ (i.e., $\mathcal{\mathcal{E}}_{2s,\delta}$
can happen). Note that not all matrices can satisfy the RIP condition
and justifying the RIP of a random matrix is also highly challenging
in general \cite{candes2008introduction}. For instance, to prove
the RIP property for conventional \emph{i.i.d.} sub-Gaussian random
matrices, the tools of concentration inequalities are needed and lots
of math derivations are usually involved to characterize the tail
probabilities of sub-Gaussian vectors \cite{baraniuk2008simple}.
On the other hand, the characterization of the RIP for the random
matrix $\Theta$ in (\ref{eq:expansion}) will be \emph{more} complicated
as $\Theta$ contains a very complicated \emph{structure}. Specifically,
the randomness of $\Theta$ comes from both the block diagonal fading
matrices $\{\mathbf{H}_{i}\}$ and the distributed compression matrices
$\{\mathbf{A}_{i}\}$. From the expression of $\Theta$ in (\ref{eq:expansion}),
we have the following two observations:
\begin{itemize}
\item \textbf{Observation I}: The columns of $\Theta$ are \emph{correlated}
due to the shared random factors of $\{\mathbf{A}_{i}(r,c)\}$. 
\item \textbf{Observation II}: The rows of $\Theta$ are \emph{correlated}
due to the shared random factors of $\{\mathbf{H}_{i}\}$. 
\end{itemize}
As such, conventional techniques \cite{baraniuk2008simple} for characterizing
the RIP for \emph{i.i.d.} sub-Gaussian random matrices cannot be applied
to our scenario. In the literature, there are a few works of RIP characterization
on structured measurement matrices, such as the Random Demodulator
matrix in \cite{BeyongN2010Baraniuk}, the Toeplitz matrix in \cite{haupt2010toeplitz},
and the random Gaussian matrix with i.i.d. row vectors while having
correlated entries within each row \cite{raskutti2010restricted}.
However, these works \cite{BeyongN2010Baraniuk,haupt2010toeplitz,raskutti2010restricted}
cannot be used in our scenario due to the different structure of $\Theta$
as described above. 

\vspace{-0.5cm}

\begin{center}
\fbox{\begin{minipage}[t]{1\columnwidth}%
Challenge 3: Analyze the RIP of $\Theta$ despite the complicated
structures in the random matrix $\Theta$ in (\ref{eq:expansion}).%
\end{minipage}}
\par\end{center}

In the following, we shall establish the RIP condition for $\Theta$.
Despite of the complicated structures above, one important feature
of $\Theta$ is that under a given realization of $\{\mathbf{H}_{i}\}$,
different \emph{columns} of $\Theta$ are \emph{conditionally independent}.
On the other hand, R. Vershynin et al. shows in \cite{vershynin2010introduction}
that random sub-Gaussian matrices with \emph{independent} \emph{columns}
is likely to satisfy the RIP with high probability under certain mild
conditions. This result \cite{vershynin2010introduction} enables
us to first deploy the conditional probability theory to \emph{decompose}
the proof of RIP of $\Theta$ into several concentration inequalities.
We then manage to prove each of these concentration inequalities by
following the derivation techniques in \cite{BeyongN2010Baraniuk}
and \cite{vershynin2010introduction}. Specifically, we show that
when the number of measurements $R$ at each RRH scales in the order
of $\mathcal{O}\left(\frac{s\log^{6}KN_{c}}{M}\right)$, the aggregate
measurement matrix $\Theta$ (with special structures) can satisfy
the RIP with high probability. 
\begin{thm}
[RIP of $\Theta$]\label{RIP-of-The-M}Suppose $KN_{c}\geq M\geq C_{1}k_{0}\log\left(\frac{eK}{k_{0}}\right)\log^{2}(KN_{c})$,
$k_{0}=\min(K,k)$. If the number of measurements $R$ at each RRH
satisfies $R\geq\frac{C_{2}\delta^{-2}k\log^{6}KN_{c}}{M}$, then
\begin{equation}
\Pr\left(\mathcal{\mathcal{E}}_{k,\delta}\right)\geq1-\frac{4}{KN_{c}},\label{eq:required_R}
\end{equation}
where $\mathcal{\mathcal{E}}_{k,\delta}$ is in (\ref{eq:event_rip})
with parameters $k$ and $\delta$, $C_{1}$ and $C_{2}$ are constants
that depend on $\delta$ and $\bar{g}$ in Assumption \ref{Channel-ModelThe-channel},
and are given in Appendix \ref{sub:Proof-of-Theorem-RIP}.\end{thm}
\begin{proof}
See Appendix \ref{sub:Proof-of-Theorem-RIP}.\end{proof}
\begin{remrk}
[Insights from Theorem \ref{RIP-of-The-M}]Note that Theorem \ref{RIP-of-The-M}
gives a \emph{sufficient} condition on the number of measurements
at the RRHs to satisfy the RIP with high probability. Combining the
result in Theorem \ref{RIP-of-The-M} with Theorem \ref{Probability-of-Correct},
we derive that in high SNR regimes (i.e., $\min_{i\in\mathcal{T}}P_{i}\gg4c_{2}^{2}N_{c}$),
when the number of measurements $R$ at each RRH scales in the order%
\footnote{Note that in the RIP result in (\ref{eq:required_R}), there is an
increase in the required number of measurements compared with the
conventional result on random matrices with i.i.d. sub-Gaussian entries
(i.e., from $\mathcal{O}\left(s\log KN_{c}\right)$ to $\mathcal{O}\left(s\log^{6}KN_{c}\right)$).
This penalty might be due to the special measurement structure in
C-RAN systems (as in Observation I-II). However, the high order of
6 on the logarithmic factors might not be necessary and it is probably
parasitic as a consequence of the techniques used to derive the theorem
(similar to \cite{BeyongN2010Baraniuk}). The empirical tests in Section
V show that a moderate number of measurements at each RRH already
leads to a good signal recovery performance. %
} of $\mathcal{O}\left(\frac{s\log^{6}KN_{c}}{M}\right)$, it suffices
to achieve a high probability of correct active user detection $\hat{\mathcal{T}}=\mathcal{T}$
(i.e., $\Pr\left(\hat{\mathcal{T}}=\mathcal{T}\right)\gtrsim1-\frac{4}{KN_{c}}$).
Based on the established RIP condition in Theorem \ref{RIP-of-The-M},
we shall further obtain a lower bound on the C-RAN performance in
Section IV. We point out that the established RIP in Theorem \ref{RIP-of-The-M}
in fact justifies the \emph{feasibility} of applying CS to achieve
distributed fronthaul compression in C-RAN with exploitation of the
uplink user sparsity. 
\end{remrk}

\begin{remrk}
[Significance of the RIP for System Design]Note that in the CS literature,
the RIP plays a \emph{central} role in characterizing the CS recovery
performance. For instance, provided that the CS measurement matrix
satisfies certain RIP condition, these state-of-the-art CS recovery
algorithms, including BP \cite{candes2005decoding}, CoSaMP \cite{needell2009cosamp}
and SP \cite{dai2009subspace}, are all shown to achieve certain performance
guarantees. Unfortunately, establishing technical conditions for RIP
is highly non-trivial and depends very much on the structure of the
measurement matrix in the compressive sensing problems. In our case,
the existing techniques for proving RIP is not applicable due to the
special structure of the measurement matrix. The proposed new proving
technique in Theorem \ref{RIP-of-The-M} can also be used to establish
RIP of other C-RAN applications. 
\end{remrk}

\section{Tradeoff Analysis Between Uplink Capacity and Fronthaul Loading in
the C-RAN}

In this section, we shall quantify the average uplink capacity of
the proposed distributed compression and joint recovery scheme in
the C-RAN. From the results, we derive simple tradeoff results between
the uplink communication performance and the fronthaul loading in
C-RAN systems. 

\vspace{-0.5cm}

\begin{center}
\fbox{\begin{minipage}[t]{1\columnwidth}%
Challenge 4: Analyze the tradeoff relationship between the C-RAN performance
and the fronthaul loading.%
\end{minipage}}
\par\end{center}

Suppose the transmit SNR of the active users are the same, i.e., $P_{i}=P$,
$\forall i\in\mathcal{T}$ for simple results. With a ZF receiver,
as in Step 3 of Algorithm 2, the recovered signal $\hat{\mathbf{x}}_{\hat{\mathcal{T}}}$
from the set of UEs $\hat{\mathcal{T}}$ can be expressed by
\begin{equation}
\hat{\mathbf{x}}_{\hat{\mathcal{T}}}=(\mathbf{\Theta}_{\hat{\mathcal{T}}})^{\dagger}\mathbf{z}=\mathbf{x}_{\hat{\mathcal{T}}}+\underset{(a)}{\underbrace{(\mathbf{\Theta}_{\hat{\mathcal{T}}})^{\dagger}\mathbf{\Theta}_{\mathcal{T}\backslash\hat{\mathcal{T}}}\mathbf{x}_{\mathcal{T}\backslash\hat{\mathcal{T}}}}}+\underset{(b)}{\underbrace{(\mathbf{\Theta}_{\hat{\mathcal{T}}})^{\dagger}\mathbf{n}}},\label{eq:decoded_model}
\end{equation}
where $(a)$ can be regarded as the interference brought by the detection
error of the active UEs and $(b)$ is the aggregate noise. Based on
(\ref{eq:decoded_model}), the sum capacity $R_{sum}$ \cite{jiang2011performance}
with the ZF receiver, is
\begin{equation}
R_{sum}=\sum_{l=1}^{\#|\hat{\mathcal{T}}|}\log\left(1+\frac{P_{l}}{\alpha_{l}}\right),\label{eq:rate}
\end{equation}
 where $P_{l}=P$ if the $l$-th element of $\hat{\mathcal{T}}$ belongs
to $\mathcal{T}$ and $P_{l}=0$ otherwise, and $\alpha_{l}$ is the
$l$-th diagonal element of the interference plus noise covariance
matrix $\Psi$, 
\begin{equation}
\Psi=\mathbf{\Theta}_{\hat{\mathcal{T}}}^{\dagger}\left(P\mathbf{\Theta}_{\mathcal{T}\backslash\hat{\mathcal{T}}}\mathbf{\Theta}_{\mathcal{T}\backslash\hat{\mathcal{T}}}^{H}+\mathbf{A}\mathbf{A}^{H}\right)(\mathbf{\Theta}_{\hat{\mathcal{T}}}^{\dagger})^{H},\label{eq:covariance}
\end{equation}
where $\mathbf{A}=\textrm{diag}\left([\begin{array}{ccc}
\mathbf{A}_{1} & \cdots & \mathbf{A}_{M}\end{array}]\right)$ and $\mathbf{A}\mathbf{A}^{H}$ is the covariance matrix of the colored
noise $\mathbf{n}$ in (\ref{eq:aggregate_noise}). Note that when
the information of the active UEs is correct, i.e., $\hat{\mathcal{T}}=\mathcal{T}$,
we obtain $P_{l}=P$, $\forall l=1,...,\#|\hat{\mathcal{T}}|$ and
$\Psi$ in (\ref{eq:covariance}) can be reduced to
\begin{equation}
\Psi=(\mathbf{\Theta}_{\hat{\mathcal{T}}}^{\dagger})\left(\mathbf{A}\mathbf{A}^{H}\right)(\mathbf{\Theta}_{\hat{\mathcal{T}}}^{\dagger})^{H}.\label{eq:simplified}
\end{equation}
Based on (\ref{eq:rate})-(\ref{eq:simplified}) and by using the
probability results regarding correct active user detection in Theorem
\ref{Probability-of-Correct}, we obtain the following bound on the
\emph{average} uplink capacity $\mathbb{E}(R_{sum})$ of C-RAN systems
(with respect to the randomness of the multiaccess channel $\{H_{ik}^{[c]}:\forall i,k,c\}$
and the local compression matrices $\{\mathbf{A}_{i}:\forall i\}$). 
\begin{thm}
[Bound of Average C-RAN Capacity]\label{Bound-of-Average}Denote
the compression rate at the RRHs as $\alpha$, i.e., $\alpha\triangleq\frac{R}{N_{c}}$.
The average capacity $\mathbb{E}(R_{sum})$ satisfies the following
bound:\end{thm}
\begin{itemize}
\item \textbf{Upper Bound}: $\mathbb{E}(R_{sum})\leq s\log\left(1+M\alpha P\right)$. 
\item \textbf{Lower Bound}: Under high SNR condition, i.e., $P\gg4c_{2}^{2}N_{c}$
and supposing the parameter $\lambda$ in Problem $\mathcal{P}_{1}$
is given by $\lambda=\left(\frac{PN_{c}}{4c_{2}^{2}}\right)^{\frac{1}{4}}$,
$\mathbb{E}(R_{sum})$ satisfies 
\begin{equation}
\mathbb{E}(R_{sum})\geq\Pr\left(\mathcal{\mathcal{E}}_{2s,\delta}\right)\cdot s\log\left(1+(1-\delta)M\alpha P\right),\label{eq:ergodic_capacity_bound}
\end{equation}

\end{itemize}
where $\Pr\left(\mathcal{\mathcal{E}}_{2s,\delta}\right)$ denotes
the probability that event $\mathcal{\mathcal{E}}_{2s,\delta}$ in
(\ref{eq:event_rip}) happens with parameters $k=2s$ and $\delta<\sqrt{2}-1$.
\begin{proof}
See Appendix \ref{sub:Proof-of-Theorem-rate}.
\end{proof}

From Theorem \ref{RIP-of-The-M} and \ref{Bound-of-Average}, we further
have the following Corollary. 
\begin{cor}
[Tradeoff Results]\label{Closed-form-Bounds-of}Consider a high
SNR scenario, i.e., $P\gg4c_{2}^{2}N_{c}$ and suppose $KN_{c}\geq M\geq C_{1}k_{0}\log\left(\frac{eK}{k_{0}}\right)\log^{2}(KN_{c})$,
where $k_{0}=\min(K,2s)$ and $\lambda$ in Problem $\mathcal{P}_{1}$
is given by $\lambda=\left(\frac{PN_{c}}{4c_{2}^{2}}\right)^{\frac{1}{4}}$
and. If the compression ratio $\alpha\triangleq\frac{R}{N_{c}}$ at
each RRH satisfies $\alpha\geq\frac{2C_{2}\delta^{-2}s\log^{6}KN_{c}}{MN_{c}}$
with $\delta<\sqrt{2}-1$, then the average sum capacity $\mathbb{E}(R_{sum})$
satisfies
\begin{eqnarray}
\left(1-\frac{4}{KN_{c}}\right)s\log\left(1+M(1-\delta)\alpha P\right) & \leq\label{eq:average_throughput_bound}\\
\mathbb{E}(R_{sum})\leq s\log\left(1+M\alpha P\right)\nonumber 
\end{eqnarray}
where $c_{2}$ is in Lemma \ref{Robust-Recovery-under}, and $C_{1}$
and $C_{2}$ are in Theorem \ref{RIP-of-The-M}.\end{cor}
\begin{proof}
(Sketch) The upper bound follows directly from Theorem \ref{Bound-of-Average}
and the lower bound is obtained by substituting $\Pr\left(\mathcal{\mathcal{E}}_{2s,\delta}\right)\geq1-\frac{4}{KN_{c}}$
in Theorem \ref{RIP-of-The-M} into (\ref{eq:ergodic_capacity_bound}).
\end{proof}

\begin{remrk}
[Interpretation of Corollary \ref{Closed-form-Bounds-of}]Note that
Corollary \ref{Closed-form-Bounds-of} gives a closed-form lower/upper
bounds on the average throughput $\mathbb{E}(R_{sum})$ under certain
requirements of the compression ratio $\alpha$. From (\ref{eq:average_throughput_bound}),
both the lower bound and upper bound of the average capacity are in
the order of $\mathcal{O}\left(s\log\left(M\alpha P\right)\right)$.
Therefore, when the compression ratio $\alpha\triangleq\frac{R}{N_{c}}$
at each RRH satisfies the order of $\mathcal{O}\left(\frac{s\log^{6}KN_{c}}{MN_{c}}\right)$,
it is sufficient to achieve a sum throughput of $\mathcal{O}\left(s\log\left(M\alpha P\right)\right)$
in the uplink C-RAN. Furthermore, from (\ref{eq:average_throughput_bound}),
we observe that the distributed fronthaul compression $\alpha$ causes
a capacity loss in terms of receiving SNR reduction with a factor
of the compression rate $\alpha$ on the fronthaul. This result uncovers
the simple tradeoff relationship between the communication performance
and the fronthaul loading in the uplink C-RAN. 
\end{remrk}

\section{Numerical Results}

In this section, we verify the effectiveness of the proposed signal
recovery and distributed fronthaul compression scheme via simulations.
Specifically, the following baselines will be considered for performance
benchmarks: 
\begin{itemize}
\item \textbf{Baseline 1} \emph{(MMSE Receiver \cite{jiang2011performance})}:
The BBUs jointly recover the transmitted signal $\mathbf{x}$ using
conventional MMSE Multi-user Detection \cite{jiang2011performance}.
\item \textbf{Baseline 2} \emph{(Separate MMSE Receiver)}: Each UE is associated
with the RRH with the largest large-scale fading gain. Each RRH \emph{separately}
recovers the transmitted signal for the set of the associated UEs
using MMSE Multi-user Detection \cite{jiang2011performance}.
\item \textbf{Baseline 3} \emph{(OMP-based ZF \cite{tropp2007signal})}:
Instead of using a convex relaxation approach, as in Algorithm 2,
the BBU pools choose to detect the set of active UEs using the OMP
\cite{tropp2007signal}, and then apply the ZF receiver to recover
the transmitted signals\emph{. }
\item \textbf{Baseline 4} \emph{(Genie-Aided} \emph{ZF)}: The set of active
UEs $\mathcal{T}$ is \emph{known} at the BBU pools so that the BBU
pools directly apply ZF on the set of active UEs $\mathcal{T}$. Note
that this serves as a performance upper on the proposed scheme. 
\end{itemize}

Consider a C-RAN system with $M=40$ single-antenna RRHs and a total
of $KN_{c}=640$ single-antenna mobile UEs being served on $N_{c}=32$
subcarriers (i.e., each subcarrier is allocated to $K=20$ UEs). Suppose
the RRHs and the mobile UEs are randomly and evenly distributed in
the circular region with radius 2km. The path loss between the UEs
and distributed RRHs are generated using the standard Log-distance
path loss model, with path loss exponent $2.5$. Denote the number
of active UEs as $s$, the number of measurements at each RRH as $R$,
and the transmit SNR at the active UEs as $P$. The threshold parameter
$\lambda$ in Algorithm 2 is set as $\lambda=\sqrt{2N_{c}}$. We further
consider that the transmitted signals on the fronthaul are quantized
with $b=10$ quantization bits per dimension ($b=10$ for each complex
number).

\subsection{Throughput Versus Compression Rate}

Figure \ref{fig:Q_Throughput-versus-bits} illustrates the \emph{per
active-user} throughput (i.e., $\frac{1}{s}\mathbb{E}(R_{sum})$)
of the C-RAN versus the fronthaul loading (in terms of quantization
bits per fronthaul link) under the number of active UEs $s=64$ and
transmit SNR $P=20$ dB. From this figure, we observe that the throughput
gets larger as the fronthaul loading increases. In addition, the proposed
scheme approaches the performance of the Genie-aided ZF as the fronthaul
loading increases. This is because as fronthaul loading (measurements
on the fronthaul) gets larger, the probability of satisfying the RIP
gets larger and the probability of correct active user detection also
gets larger. Consequently, the proposed scheme would approach the
Genie-aided ZF.

\subsection{Throughput Versus Transmit SNR}

Figure \ref{fig:Q_Throughput-versus-SNR} further illustrates the
\emph{per active-user} throughput of the C-RAN system versus the transmit
SNR $P$ under number of active UEs $s=64$ and per fronthaul bits
$B=60$. From this figure, we observe that the throughput of the proposed
scheme gets larger as the transmit SNR increases, and it performs
the same as the Genie-aided ZF in the high SNR regimes. This is because
the probability of correct detection of the active UEs gets larger
as SNR increases, as indicated in Theorem \ref{Probability-of-Correct}.
Note that the proposed upper bound in Thm. 4 can still act as a performance
upper bound in cases of fronthaul quantization as quantization would
lead to performance degradations.

\subsection{Throughput Versus Number of Active UEs}

Figure \ref{fig:Q_Throughput-versus-Sparsity} further illustrates
the \emph{per active-user} throughput of the C-RAN system versus the
number of active UEs under transmit SNR $P=20$ dB and per fronthaul
bits $B=60$. From this figure, we observe that the proposed scheme
achieves the same performance as the Genie-aided ZF in the regime
with a small number of active UEs. This is because nearly perfect
detection of the active UEs is achieved in small sparsity regimes.
This comparison further highlights the importance of correct active
user detection in C-RAN systems.

\begin{figure}
\begin{centering}
\includegraphics[scale=0.45]{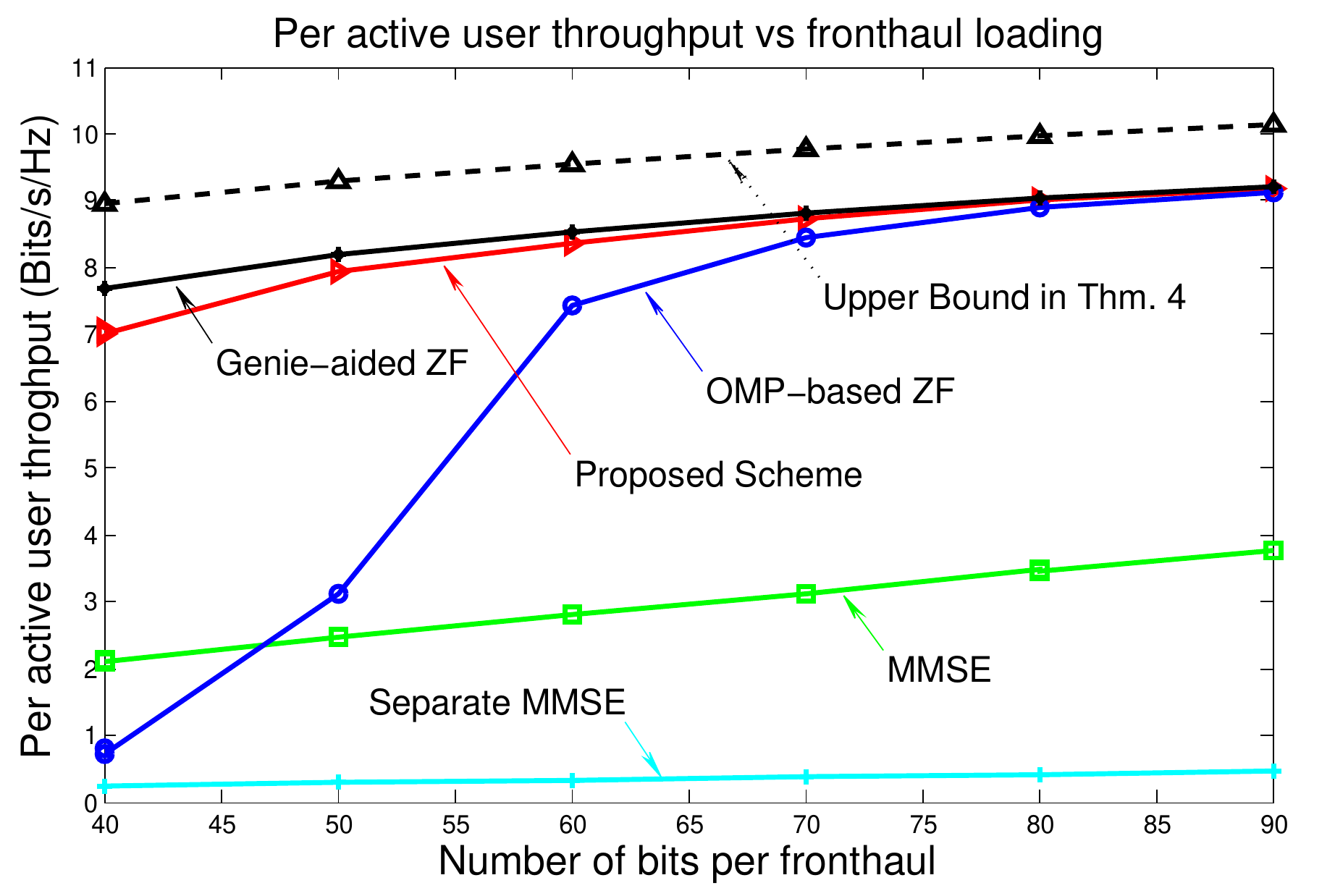}
\par\end{centering}

\protect\caption{\label{fig:Q_Throughput-versus-bits}Throughput versus the per fronthaul
loading in terms of the number of bits $B$ under transmit SNR $P=20$
dB and number of active UEs $s=64$.}
\end{figure}

\begin{figure}
\begin{centering}
\includegraphics[scale=0.45]{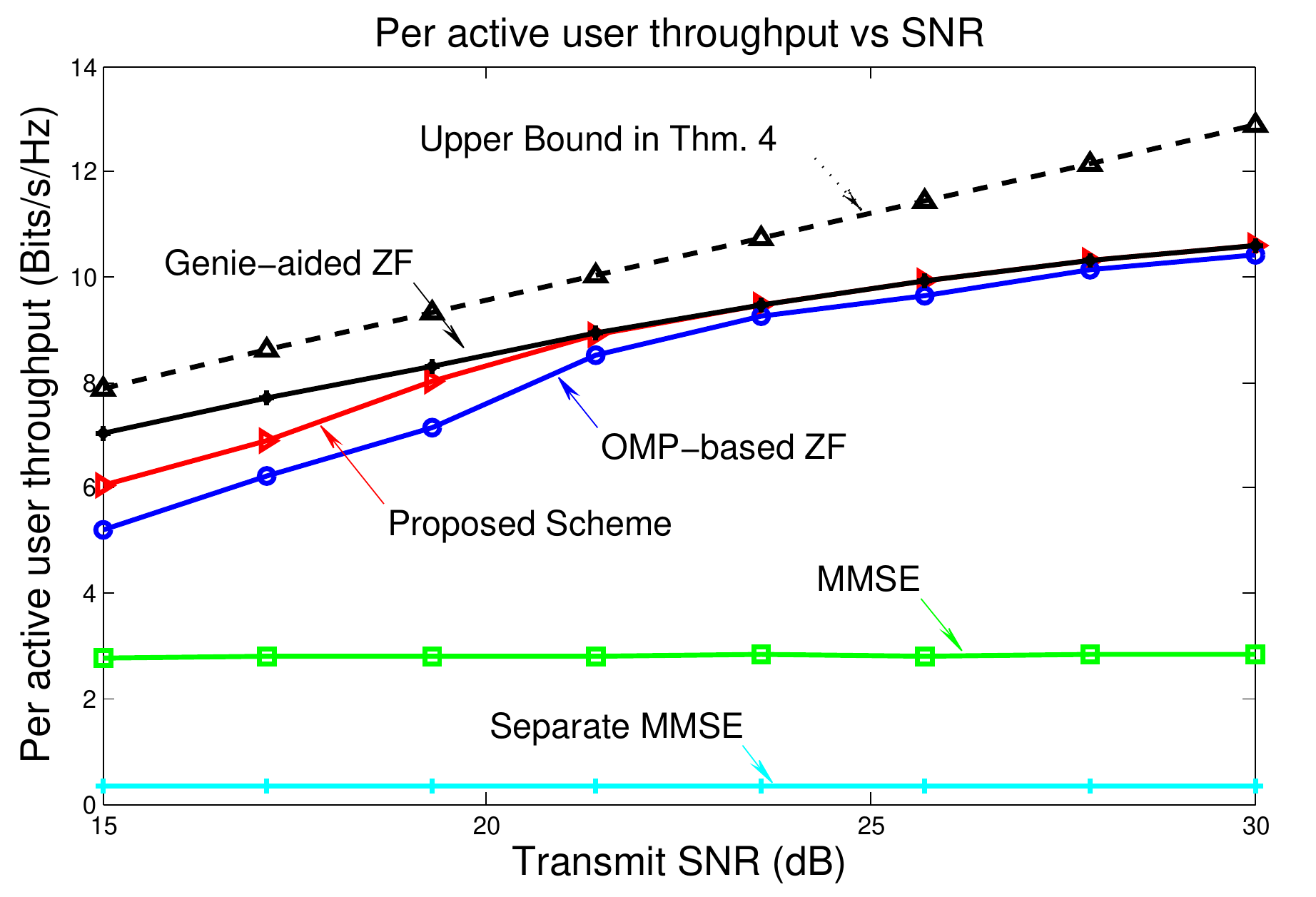}
\par\end{centering}

\protect\caption{\label{fig:Q_Throughput-versus-SNR}Throughput versus the transmit
SNR $P$ under per fronthaul quantization bits $B=60$ and number
of active UEs $s=64$.}
\end{figure}

\begin{figure}
\begin{centering}
\includegraphics[scale=0.45]{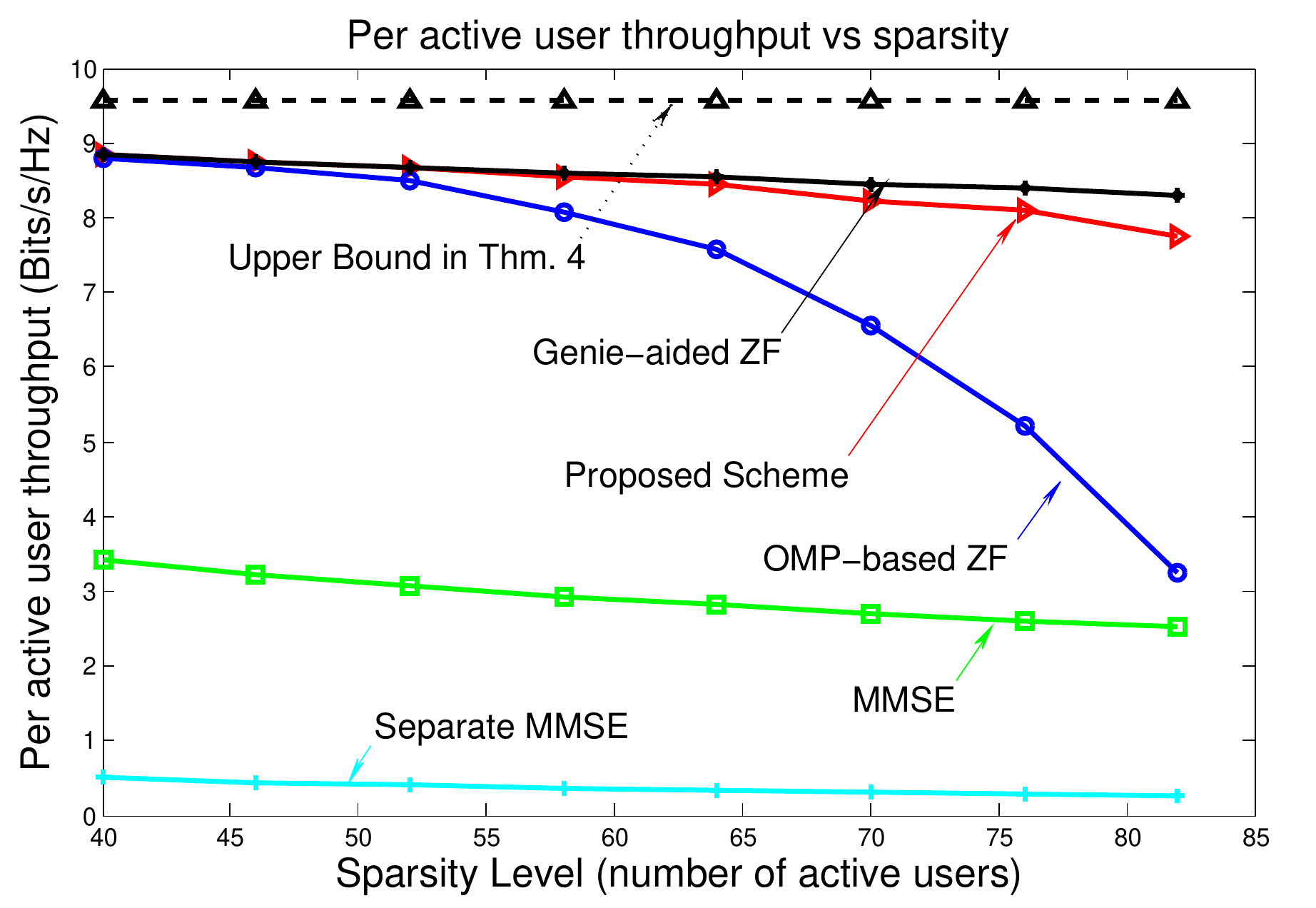}
\par\end{centering}

\protect\caption{\label{fig:Q_Throughput-versus-Sparsity}Throughput versus the number
of active UEs $s$ under per fronthaul quantization bits $B=60$ and
transmit SNR $P=20$ dB. }
\end{figure}

\section{Conclusion}

In this paper, we apply CS techniques to uplink C-RAN systems to achieve
distributed fronthaul compression with exploitation of UE signal sparsity.
We incorporate multi-access fading in C-RAN system into the CS formulation
and we conduct an end-to-end recovery of the transmitted signals from
the users. We show that the aggregate measurement matrix in the C-RAN,
which contains both the distributed compression and multiaccess fading,
can still satisfy the \emph{restricted isometry property} with high
probability. This result provides the foundation to apply CS to the
uplink C-RAN systems. Based on this, we further analyze the probability
of correct active user detection and quantify the tradeoff relationship
between the uplink capacity and the distributed fronthaul loading
in C-RAN. 

\appendix

\subsection{\label{sub:Proof-of-auxililary-lemma}Proof of Theorem \ref{auxiliary_lemma}}

We first prove the first item in Theorem \ref{auxiliary_lemma}. i)
First, the $s$ indices in $\mathcal{T}$ must be arranged at the
front in the sorted indices $\{i_{1},..,i_{KN_{c}}\}$ in Step 2 of
Algorithm 2. Otherwise, i.e., there exists a $i\in\mathcal{T}$, $j\notin\mathcal{T}$
where $|\hat{\mathbf{x}}(j)|\geq|\hat{\mathbf{x}}(i)|$, we would
obtain
\[
\left\Vert \mathbf{x}-\hat{\mathbf{x}}\right\Vert \geq\sqrt{|\mathbf{x}(i)-\hat{\mathbf{x}}(i)|^{2}+\left|0-\hat{\mathbf{x}}(j)\right|^{2}}\geq\frac{1}{\sqrt{2}}|\mathbf{x}(i)|>c_{2}\lambda
\]
which contradicts (\ref{eq:error_bound}). ii) Second, the greedy
selection of $\hat{\mathcal{T}}$ cannot stop with $\#\left|\hat{\mathcal{T}}\right|<s$.
Otherwise, i.e., there must exist an $i\in\mathcal{T}$ and $i\notin\hat{\mathcal{T}}$.
we obtain {\small{}
\begin{eqnarray*}
\left\Vert \left(\mathbf{I}-\Theta_{\hat{\mathcal{T}}}\Theta_{\hat{\mathcal{T}}}^{\dagger}\right)\mathbf{z}\right\Vert _{F} & \geq & \left\Vert \left(\mathbf{I}-\Theta_{\hat{\mathcal{T}}}\Theta_{\hat{\mathcal{T}}}^{\dagger}\right)\Theta_{\mathcal{T}}\mathbf{x}_{\mathcal{T}}\right\Vert _{F}-||\mathbf{n}||_{F}\\
 & \overset{(e_{1})}{\geq} & \sqrt{1-\delta}\left|\mathbf{x}(i)\right|-\lambda>\lambda,
\end{eqnarray*}
}which contradicts the stopping criterion in Step 2 of Algorithm 2,
where $(e_{1})$ uses the RIP property in $\mathcal{\mathcal{E}}_{2s,\delta}$.
iii) Suppose in Step 2, $\hat{\mathcal{T}}$ selects $s$ indice,
i.e., $\mathcal{T}=\hat{\mathcal{T}}$. Then $\left\Vert \left(\mathbf{I}-\Theta_{\mathcal{T}}\Theta_{\mathcal{T}}^{\dagger}\right)\mathbf{z}\right\Vert _{F}\leq||\mathbf{n}||\leq\lambda$
and Step 2 of Algorithm 2 will stop. Therefore, the first item in
Theorem \ref{auxiliary_lemma} is proved. We then prove the second
item in Theorem \ref{auxiliary_lemma}. From $\mathcal{T}=\hat{\mathcal{T}}$,
we obtain the final recovered $\hat{\mathbf{x}}$ from Algorithm 2
satisfies
\begin{align*}
\left\Vert \mathbf{x}-\hat{\mathbf{x}}\right\Vert _{F} & =\left\Vert \mathbf{x}_{\mathcal{T}}-\Theta_{\mathcal{T}}^{\dagger}\left(\Theta_{\mathcal{T}}\mathbf{x}_{\mathcal{T}}+\mathbf{n}\right)\right\Vert _{F}\\
 & =\left\Vert \Theta_{\mathcal{T}}^{\dagger}\mathbf{n}\right\Vert _{F}\overset{(e_{2})}{\leq}\frac{1}{\sqrt{1-\delta}}\left\Vert \mathbf{n}\right\Vert 
\end{align*}
where $(e_{2})$ uses the RIP property in $\mathcal{\mathcal{E}}_{2s,\delta}$.

\subsection{\label{sub:Proof-of-Lemma-Bounded}Proof of Lemma \ref{Probability-of-Bounded-noise}}

Denote the singular value decomposition (SVD) of $\mathbf{A}_{i}$
as $\mathbf{A}_{i}=\mathbf{U}_{i}\mathbf{\Sigma}_{i}\mathbf{V}_{i}^{H},$
where $\mathbf{U}_{i}\in\mathbb{C}^{R\times R}$ and $\mathbf{V}_{i}\in\mathbb{C}^{N_{c}\times R}$
are unitary matrices, i.e., $\mathbf{U}_{i}^{H}\mathbf{U}_{i}=\mathbf{I}$,
$\mathbf{V}_{i}^{H}\mathbf{V}_{i}=\mathbf{I}$ and $\mathbf{\Sigma}_{i}=\textrm{diag}\left([\beta_{i1}\cdots\beta_{iR}]\right)$
contain the singular values $\{\beta_{ij}\geq0:\forall j\}$. We first
have $\sum_{j=1}^{R}\beta_{ij}^{2}=||\mathbf{A}_{i}||_{F}^{2}=\frac{N_{c}}{M}$.
From i) $\mathbf{n}_{i}$ is standard complex Gaussian distributed,
ii) the unitary invariance property of complex Gaussian distribution,
we obtain $\mathbf{\tilde{n}}_{i}=\mathbf{V}_{i}^{H}\mathbf{n}_{i}\in\mathbb{C}^{R\times1}$
is also standard complex Gaussian distributed and is independent of
$\mathbf{V}_{i}$. From (\ref{eq:aggregate_noise}), we have 
\begin{equation}
||\mathbf{n}||_{F}^{2}=\sum_{i=1}^{M}||\mathbf{A}_{i}\mathbf{n}_{i}||_{F}^{2}=\sum_{i=1}^{M}||\mathbf{\Sigma}_{i}\mathbf{\tilde{n}}_{i}||_{F}^{2}=\frac{1}{2}\sum_{i=1}^{M}\sum_{r=1}^{R}\beta_{ij}^{2}\chi_{2}^{[ir]}\label{eq:n_expression}
\end{equation}
where $\{\chi_{2}^{[ir]}:i=1,..,M,\, j=1,..,R\}$ are i.i.d. chi-square
distributed variables with 2 degrees of freedom. We first introduce
the tools of sub-exponential variables from \cite{PeterNotes}. 
\begin{definitn}
[Sub-Exponential Variable]\label{Sub-Exponential-Variable}$X$
is \emph{sub-exponential} with parameters $(\sigma^{2},b)$ if 
\[
\ln\left(\mathbb{E}\left[\exp\left(t(X-\mathbb{E}X)\right)\right]\right)\leq\frac{t^{2}\sigma^{2}}{2},\quad\forall|t|<1/b.
\]
\hfill \IEEEQED
\end{definitn}

Based on the above definition, we further introduce the following
properties for sub-exponential variables \cite{PeterNotes}.
\begin{lemma}
[Properties of Sub-Exponential Variable]\label{Sum-of-Sub-Exponential}The
following properties holds

(i) If $X$ is \emph{sub-exponential} with parameters $(\sigma^{2},b)$,
then $X$ is \emph{also sub-exponential} with parameters $(\bar{\sigma}^{2},\bar{b})$,
$\forall\bar{\sigma}^{2}\geq\sigma^{2}$ and $\bar{b}\geq b$; 

(ii) For independent \emph{sub-exponential} variables $X_{i}$ with
parameter $(\sigma_{i}^{2},b_{i})$, $i=1,..,M$, the weighted sum
$X=\sum_{i=1}^{M}w_{i}X_{i}$ is also sub-exponential with parameters
$(\sigma^{2},b)$, where $\sigma^{2}=\sum_{i=1}^{M}w_{i}^{2}\sigma_{i}^{2}$
and $b=\max_{i}w_{i}b_{i}$. 

(iii) For \emph{sub-exponential} variables $X$ with parameters $(\sigma^{2},b)$,
\[
\Pr\left(X-\mathbb{E}X\geq t\right)\leq\begin{cases}
\exp\left(-\frac{t^{2}}{2\sigma^{2}}\right) & ,\quad0\leq t\leq\frac{\sigma^{2}}{b}\\
\exp\left(-\frac{t}{2b}\right) & ,\quad t>\frac{\sigma^{2}}{b}
\end{cases}.
\]
\hfill \IEEEQED
\end{lemma}

\begin{lemma}
[Sub-exponential Property of Chi-square Variable]\label{Sub-exponential-Property-of}Suppose
$\chi_{2}$ is a chi-square random variable with 2 degrees of freedom,
then $\chi_{2}$ is also a sub-exponential variable with parameters
$(\sigma^{2}=8,b=4)$.\end{lemma}
\begin{proof}
Note that $\mathbb{E}\left(\exp\left(\left(\chi_{2}-\mathbb{E}(\chi_{2})\right)t\right)\right)=\frac{1}{1-2t}\exp(-2t)\leq\exp(4t^{2})$,
$\forall|t|\leq\frac{1}{4}$. From Definition \ref{Sub-Exponential-Variable},
the Lemma is proved. 
\end{proof}

We utilize the above tools to derive the probability bound for $||\mathbf{n}||$.
From (\ref{eq:n_expression}) and the first two properties in Lemma
\ref{Sum-of-Sub-Exponential}, $||\mathbf{n}||_{F}^{2}=\frac{1}{2}\sum_{i=1}^{M}\sum_{j=1}^{R}\beta_{ij}^{2}\chi_{2}^{[ij]}$
is also a sub-exponential variable with parameters $\sigma^{2}=2\frac{N_{c}}{M}^{2}$
and $b=\frac{2N_{c}}{M}$. From Lemma \ref{Sub-Exponential-Variable}
(iii) and $\mathbb{E}(||\mathbf{n}||_{F}^{2})=N_{c}$, we obtain when
$\lambda^{2}>2N_{c}$, i.e., $t=\lambda^{2}-N_{c}>\sigma^{2}/b=N_{c}$,
\[
\Pr\left(||\mathbf{n}||\geq\lambda\right)\leq\exp\left(-\frac{M(\lambda^{2}-N_{c})}{4N_{c}}\right).
\]

\subsection{\label{sub:Proof-of-Lemma-active_bound}Proof of Theorem \ref{Probability-of-Correct}}

From Theorem \ref{auxiliary_lemma}, we have{\small{}
\begin{eqnarray*}
 &  & \Pr\left(\hat{\mathcal{T}}=\mathcal{T}\mid\mathcal{\mathcal{E}}_{2s,\delta}\right)\geq\Pr(\mathcal{\mathcal{E}}_{1}\;\mbox{and}\;||\mathbf{n}||_{F}\leq\lambda)\\
 & \geq & 1-\Pr(||\mathbf{n}||_{F}>\lambda)-\Pr(\overline{\mathcal{\mathcal{E}}_{1}})\\
 & \geq & 1-\exp\left(-c_{1}M\right)-s\cdot\left(1-\exp\left(-\frac{2\left(c_{2}\lambda\right)^{2}}{P_{\min}}\right)\right),
\end{eqnarray*}
}where $\overline{\mathcal{\mathcal{E}}_{1}}$ denotes the complement
event of $\mathcal{\mathcal{E}}_{1}$ in Theorem \ref{auxiliary_lemma}.
From $\Pr\left(\hat{\mathcal{T}}=\mathcal{T}\right)\geq\Pr\left(\hat{\mathcal{T}}=\mathcal{T}\mid\mathcal{\mathcal{E}}_{2s,\delta}\right)\Pr(\mathcal{\mathcal{E}}_{2s,\delta})$,
Theorem \ref{Probability-of-Correct} is proved.

\subsection{\label{sub:Proof-of-Theorem-RIP}Proof of Theorem \ref{RIP-of-The-M}}

We first introduce some notations and events to facilitate the proof.
We then divide the proof of Theorem \ref{RIP-of-The-M} into three
parts (i.e., Lemma \ref{Probability-Bound-I}-\ref{Probability-Bounds-III}
below). Denote the $q$-th norm of a random variable $X$ as $\mathbb{E}^{q}(X)\triangleq\left(\mathbb{E}\left(\left|X\right|^{q}\right)\right)^{\frac{1}{q}}$
and the norm operator $|||\cdot|||_{k}$ as $|||\mathbf{A}|||_{k}=\sup_{\#|\mathcal{S}|\leq k}\left\Vert \left(\mathbf{A}\right)_{\mathcal{\mathcal{S}\times\mathcal{S}}}\right\Vert $,
where $\left(\mathbf{A}\right)_{\mathcal{S}\times\mathcal{S}}$ returns
the principal submatrix with columns and row indices in $\mathcal{S}$.
Denote 
\[
Z\triangleq|||\Theta^{H}\Theta-\mathbf{I}|||_{k}
\]
 Therefore, matrix $\Theta$ has a $k$-th RIP property with RIC $\delta$
(i.e., event $\mathcal{\mathcal{E}}_{k,\delta}$) if $Z\leq\delta$
\cite{BeyongN2010Baraniuk} and to prove {\small{}$\Pr\left(\mathcal{\mathcal{E}}_{k,\delta}\right)\geq1-\frac{4}{KN_{c}}$}
in Theorem \ref{RIP-of-The-M}, it is sufficient to prove that 
\begin{equation}
\mbox{Pr}\left(Z>\delta\right)\leq\frac{4}{KN_{c}}.\label{eq:target}
\end{equation}
Denote $\varphi_{max}\triangleq\frac{1}{\sqrt{MR}}\max_{i,k,c}\left|H_{ik}^{[c]}\right|$
and  matrix $\mathbf{H}^{[c]}$ as 
\begin{equation}
\mathbf{H}^{[c]}=\frac{1}{\sqrt{M}}\left[\begin{array}{cccc}
H_{11}^{[c]} & H_{12}^{[c]} & \cdots & H_{1K}^{[c]}\\
H_{21}^{[c]} & H_{22}^{[c]} & \cdots & H_{2K}^{[c]}\\
\vdots & \vdots & \ddots\\
H_{M1}^{[c]} & H_{M2}^{[c]} & \cdots & H_{MK}^{[c]}
\end{array}\right].\label{eq:H^=00005Bc=00005D}
\end{equation}
Denote events $\mathcal{\mathcal{F}}_{1}$ and $\mathcal{\mathcal{F}}_{2}$
as: {\small{}
\begin{equation}
\mathcal{\mathcal{F}}_{1}:\quad\varphi_{max}\leq C_{3}\bar{g}\sqrt{\frac{\log(KN_{c})}{MR}}\label{eq:F_1_event}
\end{equation}
\begin{equation}
\mathcal{\mathcal{F}}_{2}:\quad\max_{c\in\{1,..,N_{c}\}}|||\left(\mathbf{H}^{[c]}\right)^{H}\mathbf{H}^{[c]}-\mathbf{I}|||_{k^{'}}\leq\frac{C_{4}\delta}{\sqrt{\log KN_{c}}},\label{eq:F_2_event}
\end{equation}
}where $C_{3}=\sqrt{2}e^{0.75}$, $e$ is the natural logarithm base,
$C_{4}$ ($C_{4}<1$) is a constant to be given later (Appendix \ref{sub:Proof-of-Lemma-II}).
Note that $\mathcal{\mathcal{F}}_{1}$ and $\mathcal{\mathcal{F}}_{2}$
depend on the channel matrices $\mathcal{H}\triangleq\{H_{ik}^{[c]}:\forall i,k,c\}$
only. Using the conditional probability property, we obtain{\small{}
\begin{equation}
\mbox{Pr}\left(Z>\delta\right)\leq\Pr\left(\overline{\mathcal{F}_{1}}\right)+\Pr\left(\overline{\mathcal{F}_{2}}\right)+\mbox{Pr}\left(Z>\delta\mid\mathcal{\mathcal{F}}_{1}\mathcal{\mathcal{F}}_{2}\right),\label{eq:combine_equation}
\end{equation}
}where $\overline{\mathcal{\mathcal{F}}}$ denotes the complement
event of $\mathcal{F}$. Based on (\ref{eq:combine_equation}), from
the results of $\Pr\left(\overline{\mathcal{F}_{1}}\right)\leq\frac{1}{KN_{c}}$,
$\Pr\left(\overline{\mathcal{F}_{2}}\right)\leq\frac{1}{KN_{c}}$
and $\mbox{Pr}\left(Z>\delta\mid\mathcal{\mathcal{F}}_{1}\mathcal{\mathcal{F}}_{2}\right)\leq\frac{2}{KN_{c}}$
in the following Lemma \ref{Probability-Bound-I}-\ref{Probability-Bounds-III}
respectively, we can obtain $\mbox{Pr}\left(Z>\delta\right)\leq\frac{4}{KN_{c}}$.
Therefore, Theorem \ref{RIP-of-The-M} is proved. 
\begin{lemma}
[Probability Bound I]\label{Probability-Bound-I}Suppose $M\leq KN_{c}$.
We have $\Pr\left(\overline{\mathcal{F}_{1}}\right)\leq\frac{1}{KN_{c}}$. \end{lemma}
\begin{proof}
From \cite{BeyongN2010Baraniuk} (equation (14)), for any $q$, {\small{}
\begin{eqnarray*}
\mathbb{E}^{q}\left(\varphi_{max}\right) & \leq & (MKN_{c})^{\nicefrac{1}{q}}\max_{i,k,c}\mathbb{E}^{q}(|g_{ik}^{[c]}h_{ik}^{[c]}|)\\
 & \leq & \frac{\bar{g}(MKN_{c})^{\nicefrac{1}{q}}}{\sqrt{MR}}\Gamma\left(\frac{q}{2}+1\right)^{\frac{1}{q}}
\end{eqnarray*}
} Selecting $q=4\log(KN_{c})$, we obtain $(MKN_{c})^{\nicefrac{1}{q}}\leq(KN_{c})^{\nicefrac{2}{q}}\leq(KN_{c})^{\frac{1}{2\log(KN_{c})}}=e^{0.5}$
(and $\Gamma\left(\frac{q}{2}+1\right)^{\frac{1}{q}}\leq\sqrt{\frac{q}{2}}=\sqrt{2\log(KN_{c})}$,
hence, 
\[
\mathbb{E}^{q}\left(\varphi_{max}\right)\leq e^{0.5}\sqrt{2}\bar{g}\sqrt{\frac{\log(KN_{c})}{MR}}.
\]
Using the Markov's inequality, we obtain{\footnotesize{}
\[
\Pr\left(\varphi_{max}>\sqrt{2}e^{0.75}\bar{g}\sqrt{\frac{\log(KN_{c})}{MR}}\right)\leq\left[\frac{\mathbb{E}^{q}\left(\varphi_{max}\right)}{e^{0.25}\mathbb{E}^{q}\left(\varphi_{max}\right)}\right]^{q}=\frac{1}{KN_{c}}.
\]
}{\small{}and hence Lemma \ref{Probability-Bound-I} is proved. }{\small \par}
\end{proof}

\begin{lemma}
[Probability Bounds II]\label{Probability-Bounds-II}Denote $k^{'}=\min(K,k)$.
If $M\geq C_{1}\delta^{-2}k^{'}\log(KN_{c})\log(2KN_{c})\cdot\log\left(\frac{eK}{k^{'}}\right)$,
then 
\begin{equation}
\Pr\left(\overline{\mathcal{F}_{2}}\right)\leq\frac{1}{KN_{c}},\label{eq:target_lemma6}
\end{equation}
where $C_{1}$ is a constant that depends on $C_{4}$ and $L=\sqrt{2}\bar{g}\sup_{p\geq1}2^{\frac{1}{p}}p^{-\frac{1}{2}}\Gamma\left(\frac{p}{2}+1\right)^{\frac{1}{p}}$
($L$ is a bounded constant).\end{lemma}
\begin{proof}
See Appendix \ref{sub:Proof-of-Lemma-one}. 
\end{proof}

\begin{lemma}
[Probability Bounds III]\label{Probability-Bounds-III}If $MR\geq C_{2}\delta^{-2}k\log^{6}KN_{c}$,
then 
\begin{equation}
\mbox{Pr}\left(Z>\delta\mid\mathcal{\mathcal{F}}_{1}\mathcal{\mathcal{F}}_{2}\right)\leq\frac{2}{KN_{c}},\label{eq:target_lemma7}
\end{equation}
where $C_{2}$ is a constant that depends on $\bar{g}$.\end{lemma}
\begin{proof}
See Appendix \ref{sub:Proof-of-Lemma-II}. 
\end{proof}

\subsection{\label{sub:Proof-of-Lemma-one}Proof of Lemma \ref{Probability-Bounds-II}}

Using the probability union bound, to prove (\ref{eq:target_lemma6})
in Lemma \ref{Probability-Bounds-II}, it is sufficient to prove that
{\small{}
\begin{equation}
\Pr\left(|||\left(\mathbf{H}^{[c]}\right)^{H}\mathbf{H}^{[c]}-\mathbf{I}|||_{k^{'}}\geq\frac{C_{4}\delta}{\sqrt{\log KN_{c}}}\right)\leq\frac{1}{KN_{c}^{2}}\label{eq:target_1}
\end{equation}
}holds for each index $c\in\{1,..,N_{c}\}$. We omit the subscript
of $c$ for conciseness and prove (\ref{eq:target_1}) in the following.
Let $\mathbf{h}_{k}$ denote the $k$-th column of $\mathbf{H}$,
$\forall k=1,..,K$. Denote $\mathcal{F}_{3}$ as the following event:
\[
\mathcal{F}_{3}:\quad\max_{k\in\{1,..,K\}}\left|||\mathbf{h}_{k}||^{2}-1\right|\leq\frac{C_{4}\delta}{4\sqrt{\log KN_{c}}},
\]
Using the conditional probability property, the following two equations
would be sufficient to derive (\ref{eq:target_1}). {\small{}
\begin{equation}
\Pr\left(\overline{\mathcal{F}_{3}}\right)\leq\frac{1}{2KN_{c}^{2}}.\label{eq:result_o}
\end{equation}
\begin{equation}
\Pr\left(|||\mathbf{H}^{H}\mathbf{H}-\mathbf{I}|||_{k^{'}}\geq\frac{C_{4}\delta}{\sqrt{\log KN_{c}}}\mid\mathcal{F}_{3}\right)\leq\frac{1}{2KN_{c}^{2}}.\label{eq:result_k}
\end{equation}
}{\small \par}

\subsubsection{Proof of equation (\ref{eq:result_o})}

Denote $\mathbf{h}_{k}=\frac{1}{\sqrt{M}}[\begin{array}{cccc}
g_{1k}h_{1k}, & g_{2k}h_{2k}, & \cdots & ,g_{Mk}h_{Mk}\end{array}]^{T}$. Then $||\mathbf{h}_{k}||^{2}=\frac{1}{M}\sum_{i=1}^{M}|g_{ik}h_{ik}|^{2}=\frac{1}{2M}\sum_{i}\left|g_{ik}\right|^{2}\chi_{2}^{[i]}$,
where $\{\chi_{2}^{[i]}:\forall i\}$ are i.i.d. chi-square random
variables with 2 degrees of freedom, and $\mathbb{E}\left(||\mathbf{h}_{k}||^{2}\right)=1$.
From Lemma \ref{Sub-exponential-Property-of} and Lemma \ref{Sum-of-Sub-Exponential}
(i-ii), $||\mathbf{h}_{k}||^{2}=\frac{1}{2M}\sum_{i}\left|g_{ik}\right|^{2}\chi_{2}^{[i]}$
is sub-exponential with parameters $(\sigma^{2}=\frac{2\bar{g}^{2}}{M},b=\frac{2g^{2}}{M})$.
From Lemma \ref{Sum-of-Sub-Exponential} (iii), we obtain{\small{}
\[
\Pr\left(||\mathbf{h}_{k}||^{2}-1\geq\frac{C_{4}\delta}{4\sqrt{\log KN_{c}}}\right)\leq\exp\left(-\frac{C_{4}\delta^{2}M}{64\bar{g}^{2}\log KN_{c}}\right).
\]
}Repeating the above derivation for $-||\mathbf{h}_{k}||^{2}$ instead
of $||\mathbf{h}_{k}||^{2}$, we obtain the same bound for $\Pr\left(||\mathbf{h}_{k}||^{2}\vphantom{-\mathbb{E}||\mathbf{h}_{k}||^{2}\leq-\frac{C_{4}\delta}{4\sqrt{\log KN_{c}}}}\right.$
$\left.-1\leq-\frac{C_{4}\delta}{4\sqrt{\log KN_{c}}}\right)$. Taking
a union bound over different $k\in\{1,..,K\}$ and choosing $M$ as
$M\geq C_{5}\delta^{-2}\log(KN_{c})\log(2KN_{c})$ with the constant
$C_{5}=128\bar{g}^{2}$, we derive equation (\ref{eq:result_o}).

\subsubsection{Proof of equation (\ref{eq:result_k})}

We further introduce the following tools on sub-Gaussian variables. 
\begin{definitn}
[SubGaussian Variable \cite{vershynin2010introduction}]A random
variable $X$ is called \emph{sub-Gaussian} if there exists a constant
$C$ such that 
\begin{equation}
p^{-\frac{1}{2}}\left(\mathbb{E}|X|^{p}\right)^{\frac{1}{p}}\leq C<\infty,\quad\forall p.\label{eq:sub-gaussian}
\end{equation}
Furthermore, the \emph{sub-Gaussian norm} of $X$, denoted as $||X||_{\psi}$,
is defined to be the smallest $L$ such that (\ref{eq:sub-gaussian})
holds, i.e., $||X||_{\psi}\triangleq\sup_{p\geq1}p^{-\frac{1}{2}}\left(\mathbb{E}|X|^{p}\right)^{\frac{1}{p}}$. 
\end{definitn}

\begin{definitn}
[SubGaussian Vector \cite{vershynin2010introduction}]A vector $\mathbf{a}\in\mathbb{C}^{m\times1}$
is called a sub-Gaussian vector if $(\mathbf{a}^{H}\mathbf{x})$ is
a sub-Gaussian variable for all $||\mathbf{x}||=1$, $\mathbf{x}\in\mathbb{C}^{m\times1}$.
Furthermore, the sub-Gaussian norm of $\mathbf{a}$ is defined as
$||\mathbf{a}||_{\psi}=\sup_{||\mathbf{x}||=1}||\mathbf{a}^{H}\mathbf{x}||_{\psi}$.
\hfill \IEEEQED
\end{definitn}

In \cite{vershynin2010introduction}, a restricted isometry property
(Thm 5.65) is given for matrices with SubGaussian column vectors,
where each column vector has a \emph{unit norm}. By tracking the proof,
we can obtain a more general form in which each column vector is only
required to be \emph{near} unit-normed. 
\begin{lemma}
[SubGaussian property \cite{vershynin2010introduction}]\label{Sub-gaussian-property}Let
$\mathbf{A}$ be an $m\times n$ random matrix where $\mathbf{A}(1)$,
$\mathbf{A}(2)$, ..., $\mathbf{A}(n)\in\mathbb{C}^{m\times1}$ are
its column vectors. Suppose i) $\max_{l\in\{1,..,n\}}\left|\frac{1}{m}\left\Vert \mathbf{A}(l)\right\Vert ^{2}-1\right|\leq\varepsilon$,
ii) $\mathbb{E}\left(\mathbf{A}(l)\mathbf{A}(l)^{H}\right)=\mathbf{I}_{m\times m},$
$\forall l$, and iii) the column vectors $\left\{ \mathbf{A}(l):\forall l\right\} $
are \emph{independent} sub-Gaussian vectors. If $m\geq C_{6}^{[1]}(\delta-2\varepsilon)^{-2}k^{'}\log\left(\frac{en}{k^{'}}\right)$,
then $|||\frac{1}{m}\mathbf{A}^{H}\mathbf{A}-\mathbf{I}|||_{k^{'}}>\delta$
happens with probability less than
\[
2\exp\left(-C_{6}^{[2]}(\delta-2\varepsilon)^{2}m\right)
\]
where $0\leq\epsilon\leq\frac{1}{2}\delta$, $C_{6}^{[1]}$ and $C_{6}^{[2]}$
are constants that depend on $L$ for any $L\geq\max_{l}\left\Vert \mathbf{A}(l)\right\Vert _{\psi}$.\end{lemma}
\begin{proof}
(Sketch) The result can be easily obtained by following the proof
of Thm 5.65, which is based on Thm 5. 58 in \cite{vershynin2010introduction}.
In deriving the probability bound for $|||\frac{1}{m}\mathbf{A}^{H}\mathbf{A}-\mathbf{I}|||_{k^{'}}>\delta$,
instead of substituting $\frac{1}{m}\left\Vert \mathbf{A}(l)\right\Vert ^{2}=1$
into equation (5.40) and (5.41) of \cite{vershynin2010introduction},
we shall incorporate a dynamic $1-\varepsilon\leq\frac{1}{m}\left\Vert \mathbf{A}(l)\right\Vert ^{2}\leq1+\varepsilon$
into equation (5.41) of \cite{vershynin2010introduction}. Using the
subsequent results of equation (5.41) in \cite{vershynin2010introduction}
and proposition 5.66 of \cite{vershynin2010introduction}, Lemma \ref{Sub-gaussian-property}
is obtained. 
\end{proof}

Note that when $\varepsilon=0$, Lemma \ref{Sub-gaussian-property}
is reduced to Thm 5.65 of \cite{vershynin2010introduction}. First
note that $\mathbb{E}\left(M\mathbf{h}_{k}\mathbf{h}_{k}\right)=\mathbf{I}$
and given $\mathcal{F}_{3}$, $\{\mathbf{h}_{k}:k=1,..,K\}$ are independent
sub-Gaussian vectors with the sub-Gaussian norm bounded by {\small{}
\[
L=\max_{l\in\{1,..,K\}}\left\Vert \sqrt{M}\mathbf{h}_{l}\right\Vert _{\psi}\leq\bar{g}\sup_{p\geq1}\sqrt{2}p^{-\frac{1}{2}}\Gamma\left(\frac{p}{2}+1\right)^{\frac{1}{p}}
\]
}where $\sup_{p\geq1}\sqrt{2}p^{-\frac{1}{2}}\Gamma\left(\frac{p}{2}+1\right)^{\frac{1}{p}}$
is an bounded constant \cite{vershynin2010introduction}\@. Therefore,
given $\mathcal{F}_{3}$, $\sqrt{M}\mathbf{H}$ satisfies the three
conditions of $\mathbf{A}$ in Lemma \ref{Sub-gaussian-property}.
Applying Lemma \ref{Sub-gaussian-property} with $\delta$ replaced
by $\frac{C_{4}\delta}{\sqrt{\log KN_{c}}}$ and $\varepsilon$ replaced
by $\frac{C_{4}\delta}{4\sqrt{\log KN_{c}}}$, we obtain that when
$M\geq4C_{6}^{[1]}C_{4}^{-2}\delta^{-2}\log(KN_{c})k^{'}\log\left(\frac{eK}{k^{'}}\right)$,{\footnotesize{}
\begin{equation}
\Pr\left(|||\mathbf{H}^{H}\mathbf{H}-\mathbf{I}|||_{k^{'}}\geq\frac{C_{4}\delta}{\sqrt{\log KN_{c}}}\mid\mathcal{F}_{3}\right)\leq2\exp\left(-\frac{C_{6}^{[2]}(C_{4}\delta)^{2}}{4\log KN_{c}}M\right).\label{eq:upper_bound}
\end{equation}
}From equation (\ref{eq:upper_bound}), $M\geq C_{1}\delta^{-2}k^{'}\log\left(KN_{c}\right)\log\left(2KN_{c}\right)\cdot\log\left(\frac{eK}{k^{'}}\right)$
where $C_{1}=\max\left(4C_{6}^{[1]}C_{4}^{-2},\;\frac{8}{C_{6}^{[2]}}C_{4}^{-2},\; C_{5}\right)$,
we derive equation (\ref{eq:result_k}).

\subsection{\label{sub:Proof-of-Lemma-II}Proof of Lemma \ref{Probability-Bounds-III}}

Denote the $r$-th row vector of $\Theta$ as $\mathbf{z}_{r}\in\mathbb{C}^{1\times KN_{c}}$.
We have 
\begin{equation}
Z\triangleq|||\Theta^{H}\Theta-\mathbf{I}|||_{k}=|||\sum_{r=1}^{MR}\mathbf{z}_{r}^{H}\mathbf{z}_{r}-\mathbf{I}|||_{k}.\label{eq:target_expression}
\end{equation}
To prove $\mbox{Pr}\left(Z>\delta\mid\mathcal{\mathcal{F}}_{1}\mathcal{\mathcal{F}}_{2}\right)\leq\frac{2}{KN_{c}}$
in Lemma \ref{Probability-Bounds-III}, it suffices to prove that
for any given channel realization $\mathcal{H}\triangleq\{H_{ik}^{[c]}:\forall i,k,c\}$
where $\mathcal{\mathcal{F}}_{1}\mathcal{\mathcal{F}}_{2}$ are satisfied,
the following equation is satisfied:
\begin{equation}
\mbox{Pr}\left(Z>\delta\mid\mathcal{H},\mathcal{F}_{1}\mathcal{F}_{2}\right)\leq\frac{2}{KN_{c}}.\label{eq:target_lemma_new}
\end{equation}
We then prove (\ref{eq:target_lemma_new}) below and the main proof
flow follows from \cite{BeyongN2010Baraniuk}. First of all, we introduce
the following tools (Lemma \ref{Rudelson-Vershynin}-\ref{Tail-ProbabilityLet})
from \cite{BeyongN2010Baraniuk}. 
\begin{lemma}
[Rudelson-Vershynin \cite{BeyongN2010Baraniuk}]\label{Rudelson-Vershynin}Suppose
that $\{\mathbf{z}_{r}\}$ is a sequence of $MR$ vectors in $\mathbb{C}^{KN_{c}}$,
where $MR\leq KN_{c}$, and $||\mathbf{z}_{r}||_{\infty}\leq B$ $\forall r\in\{1,..,MR\}$.
Let $\{\xi_{r}\in\{-1,+1\}\}$ be an independent Rademacher series.
Then {\small{}
\[
\mathbb{E}|||\sum_{r}\xi_{r}\mathbf{z}_{r}^{H}\mathbf{z}_{r}|||_{k}\leq\beta\left(|||\sum_{r}\mathbf{z}_{r}^{H}\mathbf{z}_{r}|||_{k}\right)^{\frac{1}{2}},
\]
}where $\beta\leq C_{7}B\sqrt{k}\log^{2}\left(KN_{c}\right)$ and
$C_{7}$ is an absolute constant. 
\end{lemma}

\begin{lemma}
[Tail Probability (\cite{BeyongN2010Baraniuk}, Proposition 19)]\label{Tail-ProbabilityLet}Let
$\{\mathbf{Y}_{r}\}$ be a sequence of independent random matrices
and $|||\mathbf{Y}_{r}|||_{k}\leq B$, $\forall r$. Let $Y=|||\sum_{r}\mathbf{Y}_{r}|||_{k}$.
Then 
\[
\mbox{Pr}\left\{ Y>C_{8}[u\mathbb{E}Y+tB]\right\} \leq e^{-u^{2}}+e^{-t}
\]
for all $u,$ $t\geq1$, where $C_{8}$ is an absolute constant.\hfill \IEEEQED
\end{lemma}

Denote $\mathbb{E}_{|\mathcal{H},\mathcal{F}_{1}\mathcal{F}_{2}}(Z)\triangleq\mathbb{E}(Z\mid\mathcal{H},\mathcal{F}_{1}\mathcal{F}_{2})$
for simplicity. Based on Lemma \ref{Rudelson-Vershynin}-\ref{Tail-ProbabilityLet},
we obtain the following two equations 
\begin{equation}
\mathbb{E}_{|\mathcal{H},\mathcal{F}_{1}\mathcal{F}_{2}}(Z)\leq\frac{C_{9}\delta}{\sqrt{\log KN_{c}}},\label{eq:first_1}
\end{equation}
\begin{equation}
\mbox{Pr}\left(Z>C_{10}\delta\mid\mathcal{H},\mathcal{F}_{1}\mathcal{F}_{2}\right)\leq\frac{2}{KN_{c}}\label{eq:first_2}
\end{equation}
where $C_{9}=\left(\sqrt{\frac{16C_{3}^{4}C_{7}^{4}\bar{g}^{4}}{C_{2}^{2}}}+\sqrt{\frac{4C_{3}^{2}C_{7}^{2}\bar{g}^{2}}{C_{2}}}+2C_{4}\right)$,
$C_{10}=\left(2C_{9}C_{8}+\left(\frac{2C_{3}^{2}\bar{g}^{2}}{C_{2}}\right)C_{8}+2C_{9}\right)$.
Note that (i) equation (\ref{eq:first_2}) can derive our target equation
(\ref{eq:first_2}) by properly selecting the constants $C_{4}$ and
$C_{2}$ (e.g., $C_{4}=\frac{1}{16(C_{8}+1)}$, $C_{2}=\max\left(64(2C_{8}+2)^{2}C_{7}^{2}C_{3}^{2}\bar{g}^{2},\;8C_{8}C_{3}^{2}\bar{g}^{2}\right)$
so that $C_{10}\leq1$; (ii) equation (\ref{eq:first_2}) is derived
based on equation (\ref{eq:first_1}){\small{}. Next we focus on proving
equation (\ref{eq:first_1}) and (\ref{eq:first_2}) in the following. }{\small \par}

\subsubsection{Proof of equation (\ref{eq:first_1})}

Using the triangular inequality, we obtain{\small{}
\begin{align}
\mathbb{E}_{|\mathcal{H},\mathcal{F}_{1}\mathcal{F}_{2}}(Z) & \leq\underset{\triangleq E_{1}}{\underbrace{|||\mathbb{E}_{|\mathcal{H},\mathcal{F}_{1}\mathcal{F}_{2}}\left(\sum_{r=1}^{MJ}\mathbf{z}_{r}^{H}\mathbf{z}_{r}\right)-\mathbf{I}|||_{k}}}+\label{eq:mat}\\
 & \underset{\triangleq E_{2}}{\underbrace{\mathbb{E}_{|\mathcal{H},\mathcal{F}_{1}\mathcal{F}_{2}}|||\sum_{r=1}^{MJ}\left(\mathbf{z}_{r}^{H}\mathbf{z}_{r}-\mathbb{E}_{|\mathcal{H},\mathcal{F}_{1}\mathcal{F}_{2}}\left(\mathbf{z}_{r}^{H}\mathbf{z}_{r}\right)\right)|||_{k}}}\nonumber 
\end{align}
}From event $\mathcal{F}_{2}$ given in (\ref{eq:F_2_event}), and
the fact that {\small{}
\[
\mathbb{E}_{|\mathcal{H},\mathcal{F}_{1}\mathcal{F}_{2}}\left(\sum_{r=1}^{MJ}\mathbf{z}_{r}^{H}\mathbf{z}_{r}\right)=\left[\begin{array}{ccc}
\ddots & \mathbf{0}\\
\mathbf{0} & \left(\mathbf{H}^{[c]}\right)^{H}\mathbf{H}^{[c]} & \mathbf{0}\\
 & \mathbf{0} & \ddots
\end{array}\right]_{c\in\{1,2,...,N_{c}\}}
\]
}We obtain $E_{1}\leq\frac{C_{4}\delta}{\sqrt{\log KN_{c}}}$. Let
$\{\xi_{r}\in\{-1,+1\}\}$ be an independent Rademacher series \cite{BeyongN2010Baraniuk}.
We obtain, 
\begin{eqnarray}
E_{2} & \overset{(a_{1})}{\leq} & 2\mathbb{E}_{|\mathcal{H},\mathcal{F}_{1}\mathcal{F}_{2}}|||\sum_{r=1}^{MJ}\xi_{r}\mathbf{z}_{r}^{H}\mathbf{z}_{r}|||_{k}\label{eq:rudeson}\\
 & \overset{(a_{2})}{\leq} & 2\mathbb{E}_{|\mathcal{H},\mathcal{F}_{1}\mathcal{F}_{2}}\left(\beta|||\sum_{r=1}^{MJ}\mathbf{z}_{r}^{H}\mathbf{z}_{r}|||^{\frac{1}{2}}\right)\nonumber \\
 & \overset{(a_{3})}{\leq} & 2\left[\mathbb{E}_{|\mathcal{H},\mathcal{F}_{1}\mathcal{F}_{2}}\left(\beta^{2}\right)\right]^{\nicefrac{1}{2}}\left(\mathbb{E}_{|\mathcal{H},\mathcal{F}_{1}\mathcal{F}_{2}}(Z)+1\right)^{\nicefrac{1}{2}}\nonumber 
\end{eqnarray}
where $(a_{1})$ uses the symmetrization property (i.e., Lemma 5.46
of \cite{vershynin2010introduction}), $(a_{2})$ uses the Rudelson-Vershynin
Lemma and $(a_{3})$ comes from the Cauch-Schwarz inequality. When
$\mathcal{F}_{1}\mathcal{F}_{2}$ happens, we obtain $||\mathbf{z}_{r}||_{\infty}\leq C_{3}\bar{g}\sqrt{\frac{\log(KN_{c})}{MR}}$,
$\forall r$ and hence {\small{}
\[
\left[\mathbb{E}_{|\mathcal{H},\mathcal{F}_{1}\mathcal{F}_{2}}\left(\beta^{2}\right)\right]^{\nicefrac{1}{2}}\leq C_{7}C_{3}\bar{g}\sqrt{\frac{k\log^{5}(KN_{c})}{MR}}\leq\frac{C_{7}C_{3}\bar{g}\delta}{\sqrt{C_{2}\log KN_{c}}}
\]
} in (\ref{eq:rudeson}). Suppose $C_{2}\geq C_{7}^{2}C_{3}^{2}\bar{g}^{2}$
so that $\frac{C_{7}C_{3}\bar{g}\delta}{\sqrt{C_{2}\log KN_{c}}}\leq\frac{\delta}{\sqrt{\log KN_{c}}}$.
From (\ref{eq:mat}) and (\ref{eq:rudeson}), we obtain {\small{}
\[
\mathbb{E}_{|\mathcal{H},\mathcal{F}_{1}\mathcal{F}_{2}}(Z)\leq\frac{C_{4}\delta}{\sqrt{\log KN_{c}}}+\frac{2C_{7}C_{3}\bar{g}\delta}{\sqrt{C_{2}\log KN_{c}}}\left(\mathbb{E}_{|\mathcal{H},\mathcal{F}_{1}\mathcal{F}_{2}}(Z)+1\right)^{\nicefrac{1}{2}}
\]
}which further derives our target equation (\ref{eq:first_1}).

\subsubsection{Proof of equation (\ref{eq:first_2})}

We prove (\ref{eq:first_2}) based on (\ref{eq:first_1}). First,
define the symmetrized random variable
\[
\mathbf{Y}_{r}=\varepsilon_{r}\left(\mathbf{z}_{r}^{H}\mathbf{z}_{r}-\mathbf{\tilde{z}}_{r}^{H}\mathbf{\tilde{z}}_{r}\right),\; Y=|||\sum_{r=1}^{MJ}\varepsilon_{r}\left(\mathbf{z}_{r}^{H}\mathbf{z}_{r}-\mathbf{\tilde{z}}_{r}^{H}\mathbf{\tilde{z}}_{r}\right)|||_{k}
\]
where $\{\varepsilon_{r}:r=1,...,MR\}$ are an i.i.d. Radermancher
sequence and $\mathbf{\tilde{z}}_{r}$ is an independent copy of $\mathbf{z}_{r}$.
Therefore, 
\[
\mathbb{E}_{|\mathcal{H},\mathcal{F}_{1}\mathcal{F}_{2}}Y\leq2\mathbb{E}_{|\mathcal{H},\mathcal{F}_{1}\mathcal{F}_{2}}\left(Z\right)\leq\frac{2C_{9}\delta}{\sqrt{\log(KN_{c})}}.
\]

From the symmetrization property (i.e., equation (6.1) in \cite{vakhania1987probability}),
we obtain {\small{}
\begin{equation}
\mbox{Pr}\left(Z>2\mathbb{E}(Z)+u\mid\mathcal{H},\mathcal{F}_{1}\mathcal{F}_{2}\right)\leq2\mathbb{P}\left(Y>u\mid\mathcal{H},\mathcal{F}_{1}\mathcal{F}_{2}\right).\label{eq:property1}
\end{equation}
}On the other hand, given $\mathcal{F}_{1}\mathcal{F}_{2}$, we have,

\begin{align*}
 & \max_{r}|||\mathbf{Y}_{r}|||_{k}=\max_{r}|||\mathbf{z}_{r}^{H}\mathbf{z}_{r}-\mathbf{\tilde{z}}_{r}^{H}\mathbf{\tilde{z}}_{r}|||_{k}\leq2\max_{r}|||\mathbf{z}_{r}^{H}\mathbf{z}_{r}|||_{k}\\
 & \leq2k\varphi_{max}^{2}\leq2C_{3}^{2}\bar{g}^{2}\frac{k\log(KN_{c})}{MR}\leq\left(\frac{2C_{3}^{2}\bar{g}^{2}}{C_{2}}\right)\frac{\delta}{\log(KN_{c})}
\end{align*}
Note that under a given realization $\mathcal{H}$, different rows
vectors $\{\mathbf{z}_{r}:\forall r\}$ are independent from (\ref{eq:expansion})
and hence $\{\mathbf{Y}_{r}:\forall r\}$ are also independent. Therefore,
the tail probability property in Lemma \ref{Tail-ProbabilityLet}
can be applied. Applying Lemma \ref{Tail-ProbabilityLet} with $u=\sqrt{\log(KN_{c})}$
and $t=\log(KN_{c})$, using equation (\ref{eq:property1}), we obtain{\small{}
\begin{eqnarray}
 &  & \mbox{Pr}\left(Z>\left(2C_{9}C_{8}+\left(\frac{2C_{3}^{2}\bar{g}^{2}}{C_{2}}\right)C_{8}+2C_{9}\right)\delta\mid\mathcal{H},\mathcal{F}_{1}\mathcal{F}_{2}\right)\nonumber \\
 & \leq & 2\mbox{Pr}\left(Y>\left(2C_{9}C_{8}+\left(\frac{2C_{3}^{2}\bar{g}^{2}}{C_{2}}\right)\right)\delta\mid\mathcal{H},\mathcal{F}_{1}\mathcal{F}_{2}\right)\nonumber \\
 & \leq & 2\mbox{Pr}\left(Y>C_{3}\left\{ u\mathbb{E}_{|\mathcal{H},\mathcal{F}_{1}\mathcal{F}_{2}}Y+t\frac{\delta}{\log(KN_{c})}\right\} \mid\mathcal{H},\mathcal{F}_{1}\mathcal{F}_{2}\right)\nonumber \\
 & \leq & 2e^{-\log(KN_{c})}=\frac{2}{KN_{c}}.\label{eq:equation}
\end{eqnarray}
}and hence equation (\ref{eq:first_2}) is obtained.

\subsection{\label{sub:Proof-of-Theorem-rate}Proof of Theorem \ref{Bound-of-Average}}

First, we can always upper bound the average data rate by the capacity
obtained in the ideal case in which the signal support of $\mathbf{x}$
is correct, i.e., $\hat{\mathcal{T}}=\mathcal{T}$. Therefore,{\small{}
\begin{equation}
\mathbb{E}(R_{sum})\leq\mathbb{E}\left(\sum_{i=1}^{s}\log\left(1+\frac{P}{\alpha_{i}}\right)\right).\label{eq:upper}
\end{equation}
}where $\{\alpha_{i}:i=1,...,s\}$ are the diagonal elements of the
Hermitian matrix $\Psi=(\mathbf{\Theta}_{\mathcal{T}})^{\dagger}\left(\mathbf{A}\mathbf{A}^{H}\right)\left((\mathbf{\Theta}_{\mathcal{T}})^{\dagger}\right)^{H}$.
Denote the SVD of $\mathbf{A}_{i}$ as $\mathbf{A}_{i}=\mathbf{U}_{i}\mathbf{\Sigma}_{i}\mathbf{V}_{i}^{H}$,
as in Appendix \ref{sub:Proof-of-Lemma-Bounded}, where $\mathbf{U}_{i}\in\mathbb{C}^{R\times R}$
and $\mathbf{V}_{i}\in\mathbb{C}^{N_{c}\times R}$ are unitary matrices
and $\mathbf{\Sigma}_{i}\in\mathbb{C}^{R\times R}$ is the diagonal
matrix with singular values. Denote {\small{}
\[
\mathbf{H}=\left[\begin{array}{ccc}
\mathbf{H}_{1}^{T} & \cdots & \mathbf{H}_{M}^{T}\end{array}\right]^{T},\quad\mathbf{V}=\textrm{diag}\left(\left[\begin{array}{ccc}
\mathbf{V}_{1} & \cdots & \mathbf{V}_{M}\end{array}\right]\right).
\]
}We obtain $\Psi=\left(\mathbf{H}_{\mathcal{T}}^{H}\mathbf{V}\mathbf{V}^{H}\mathbf{H}_{\mathcal{T}}\right)^{-1}$,
and hence $\sum\frac{1}{\alpha_{i}}=\textrm{tr}\left(\mathbf{H}_{\mathcal{T}}^{H}\mathbf{V}\mathbf{V}^{H}\mathbf{H}_{\mathcal{T}}\right)$.
From the generation method of $\mathbf{A}_{i}$ in Definition \ref{Local-Compression-Matrices},
$\mathbf{V}_{i}^{H}\mathbf{P}_{i}$ has the same distribution as $\mathbf{V}_{i}$
and is independent of $\mathbf{P}_{i}$ for any permutation matrix
$\mathbf{P}_{i}\in\mathbb{C}^{N_{c}\times N_{c}}$. Denote $\mathbf{P}=\textrm{diag}\left(\left[\begin{array}{ccc}
\mathbf{P}_{1} & \cdots & \mathbf{P}_{M}\end{array}\right]\right)$. We obtain {\small{}
\begin{equation}
\mathbb{E}\sum_{i=1}^{s}\frac{1}{\alpha_{i}}=\textrm{tr}\left(\mathbb{E}\left(\mathbf{V}\mathbf{V}^{H}\right)\mathbb{E}\left(\mathbf{P}^{H}\mathbf{H}_{\mathcal{T}}\mathbf{H}_{\mathcal{T}}^{H}\mathbf{P}\right)\right)=\frac{sMR}{N_{c}}=sM\alpha.\label{eq:equality}
\end{equation}
}Based on (\ref{eq:upper}) and (\ref{eq:equality}), from the Jensen's
inequality and the concavity of function $h(x)=\log\left(1+x\right)$,
we obtain
\[
\mathbb{E}\left(\sum_{i=1}^{s}\log\left(1+P\left(\frac{1}{\alpha_{i}}\right)\right)\right)\leq s\log\left(1+M\alpha P\right),
\]
and hence the upper bound is obtained. Next, we prove the lower bound.
From Theorem \ref{Probability-of-Correct}, when $P\gg4c_{2}^{2}N_{c}$
and $\lambda=\left(\frac{PN_{c}}{4c_{2}^{2}}\right)^{\frac{1}{4}}$,
$\Pr(\hat{\mathcal{T}}=\mathcal{T})\geq\Pr\left(\mathcal{\mathcal{E}}_{2s,\delta}\right)$
and we obtain{\small{}
\begin{align}
\mathbb{E}(R_{sum}) & \geq\Pr\left(\mathcal{\mathcal{E}}_{2s,\delta}\right)\cdot\mathbb{E}\left[\sum_{i=1}^{s}\log\left(1+\frac{P}{\alpha_{i}}\right)\mid\mathcal{\mathcal{E}}_{2s,\delta}\right]\label{eq:equation1}
\end{align}
}where $\{\alpha_{i}:i=1,...,s\}$ are the diagonal elements of the
Hermitian matrix {\small{}$\Psi=(\mathbf{\Theta}_{\mathcal{T}})^{\dagger}\left(\mathbf{A}\mathbf{A}^{H}\right)\left((\mathbf{\Theta}_{\mathcal{T}})^{\dagger}\right)^{H}$}.
Based on (\ref{eq:equation1}), to prove the lower bound in Theorem
\ref{Bound-of-Average}, it is sufficient to prove that {\small{}
\begin{equation}
\mathbb{E}\left[\sum_{i=1}^{s}\log\left(1+\frac{P}{\alpha_{i}}\right)\mid\Theta,\mathcal{\mathcal{E}}_{2s,\delta}\right]\geq s\log\left(1+(1-\delta)M\alpha P\right).\label{eq:equation2_target}
\end{equation}
}We then prove (\ref{eq:equation2_target}) in the following. Note
that given $\Theta$ and $\mathcal{\mathcal{E}}_{2s,\delta}$, the
conditional distribution of $\mathbf{P}\mathbf{A}$ is the same as
$\mathbf{A}$ for any diagonal matrix $\mathbf{P}\in\mathbb{C}^{MR\times MR}$
where the diagonal elements of $\mathbf{P}$ are randomly drawn from
\{-1, 1\} with equal probability. Therefore, we obtain{\small{}
\[
\mathbb{E}\left(\mathbf{A}\mathbf{A}^{H}\mid\Theta,\mathcal{\mathcal{E}}_{2s,\delta}\right)=\mathbb{E}\left(\mathbf{P}\mathbf{A}\mathbf{A}^{H}\mathbf{P}\mid\Theta,\mathcal{\mathcal{E}}_{2s,\delta}\right)=\frac{1}{M\alpha}\mathbf{I}.
\]
}Based on this equation, we further obtain{\small{}
\begin{align}
 & \mathbb{E}\left(\sum_{i=1}^{s}\alpha_{i}\mid\Theta,\mathcal{\mathcal{E}}_{2s,\delta}\right)=\mathbb{E}\left(\textrm{tr}(\Psi)\mid\Theta,\mathcal{\mathcal{E}}_{2s,\delta}\right)\nonumber \\
= & \textrm{tr}\left((\mathbf{\Theta}_{\mathcal{T}})^{\dagger}\mathbb{E}\left(\mathbf{A}\mathbf{A}^{H}\mid\Theta,\mathcal{\mathcal{E}}_{2s,\delta}\right)\left((\mathbf{\Theta}_{\mathcal{T}})^{\dagger}\right)^{H}\right)\nonumber \\
= & \frac{1}{M\alpha}\textrm{tr}\left((\mathbf{\Theta}_{\mathcal{T}}^{H}\mathbf{\Theta}_{\mathcal{T}})^{-1}\right)\overset{(d)}{\leq}\frac{s}{(1-\delta)M\alpha}\label{eq:new_sum}
\end{align}
}where $(d)$ comes from the fact that the $s$ eigenvalues of the
Hermitian matrix $\mathbf{\Theta}_{\mathcal{T}}^{H}\mathbf{\Theta}_{\mathcal{T}}$
are bounded in $[1-\delta,1+\delta]$ from the RIP condition (i.e.,
$\mathcal{\mathcal{E}}_{2s,\delta}$) of $\mathbf{\Theta}$. Based
on (\ref{eq:new_sum}), from the Jensen's inequality and the convexity
of function $h(x)=\log(1+\frac{1}{x})$, $\forall x>0$, we derive
(\ref{eq:equation2_target}). Therefore, the lower bound of $\mathbb{E}(R_{sum})$
is proved. 

\bibliographystyle{IEEEtran}
\bibliography{CRAN_REF}

\end{document}